%% file: IEEE_TSP_v2.tex
\documentclass[10pt,twocolumn,twoside]{IEEEtran}

\input{package_mode}
\usepackage{graphicx,mathtools}
\usepackage{paralist, stfloats}
\usepackage{booktabs, multirow}
\usepackage{psfrag, subfigure, bm}
\usepackage[usenames,dvipsnames]{pstricks}
\usepackage{pst-grad}
\usepackage[lowtilde]{url}
\input{package_IEEE}

\def\transp{\mathsf{T}}

\DeclareMathOperator*{\argmin}{\mathrm{argmin}}
\newtheorem{corollary}{Corollary}

\begin{document}

\title{Distributed Maximum Likelihood \\Sensor Network Localization}
%
%
%

\author{Andrea~Simonetto* and~Geert~Leus


\thanks{The authors are with the Faculty of Electrical Engineering, Mathematics and Computer Science, Delft University of Technology, 2826 CD Delft, The Netherlands. e-mails: \footnotesize{$\{$a.simonetto, g.j.t.leus$\}$@tudelft.nl}. * Corresponding author: Andrea Simonetto, phone: (+31)152782845, fax: (+31)152786190, e-mail: \footnotesize{a.simonetto@tudelft.nl}. This research was supported in part by STW under the D2S2 project from the ASSYS program (project 10561).}}

\maketitle

\begin{abstract}
\boldmath
We propose a class of convex relaxations to solve the sensor network localization problem, based on a maximum likelihood (ML) formulation. This class, as well as the tightness of the relaxations, depends on the noise probability density function (PDF) of the collected measurements. We derive a computational efficient edge-based version of this ML convex relaxation class and we design a distributed algorithm that enables the sensor nodes to solve these edge-based convex programs locally by communicating only with their close neighbors. This algorithm relies on the alternating direction method of multipliers (ADMM), it converges to the centralized solution, it can run asynchronously, and it is computation error-resilient. Finally, we compare our proposed distributed scheme with other available methods, both analytically and numerically, and we argue the added value of ADMM, especially for large-scale networks. 

\end{abstract}

\vskip0.2cm
\begin{center}\textbf{EDICS Category: SEN-DIST, SEN-COLB}\end{center}

\section{Introduction}\label{sec:introduction}
Nowadays, wireless sensor networks are developed to provide fast, cheap, reliable, and scalable hardware solutions to a large number of industrial applications, ranging from surveillance \cite{Biswas2006, Raty2010} and tracking \cite{Songhwai2007, Liu2007a} to exploration \cite{Sun2005, Leonard2007}, monitoring \cite{Corke2010, Sun2011}, robotics~\cite{Zhou2011}, and other sensing tasks \cite{Arampatzis2005}. From the software perspective, an increasing effort is spent on designing distributed algorithms that can be embedded in these sensor networks, providing high reliability with limited computation and communication requirements for the sensor nodes. Estimating the location of the nodes based on pair-wise distance measurements is regarded as a key enabling technology in many of the aforementioned scenarios, where GPS is often not employable. 

From a strictly mathematical standpoint, this sensor network localization problem can be formulated as determining the node position in $\mathbb{R}^2$ or $\mathbb{R}^3$ ensuring their consistency with the given inter-sensor distance measurements and (in some cases) with the location of known anchors. As it is well known, such a fixed-dimensional problem (often phrased as a polynomial optimization) is NP-hard in general. Consequently, there have been significant research efforts in developing algorithms and heuristics that can accurately and efficiently localize the nodes in a given dimension~\cite{Langendoen2003, Patwari2005, Mao2007}. Besides heuristic geometric schemes, such as multi-lateration, typical methods encompass multi-dimensional scaling~\cite{Shang2003,Cheung2005}, belief propagation techniques~\cite{Wymeersch2009}, and standard non-linear filtering~\cite{Cattivelli2010}.  

A very powerful approach to the sensor network localization problem is to use convex relaxation techniques to massage the non-convex problem to a more tractable yet approximate formulation. First adopted in~\cite{Doherty2001}, this modus operandi has since been extensively developed in the literature (see for example~\cite{Luo2010} for a comprehensive survey in the field of signal processing). Semidefinite programming (SDP) relaxations for the localization problem have been proposed in~\cite{Biswas2004, Biswas2006a, Weinberger2006, Sun2006,Wang2008,Lui2009,Pong2011,Pong2012}. Theoretical properties of these methods have been discussed in~\cite{So2007, Javanmard2011, Shamsi2012}, while their efficient implementation has been presented in~\cite{Fukuda2001, Nakata2003, Weinberger2007, Leung2009, Kim2009}. Further convex relaxations, namely second-order cone programming relaxations (SOCP) have been proposed in~\cite{Tseng2007} to alleviate the computational load of standard SDP relaxations, at the price of some performance degradation. Highly accurate and highly computational demanding sum of squares (SOS) convex relaxations have been instead employed in~\cite{Nie2009}. 




Despite the richness of the convex relaxation literature, two main aspects have been overlooked. First of all, a comprehensive characterization of these convex relaxations based on the maximum likelihood (ML) formulation is missing. In~\cite{Biswas2006a,Lui2009,Yang2009,Oguz-Ekim2010} ML-based relaxations are explored, but only for specific noise models (mainly Gaussian noise), without a proper understanding of how different noise models would affect performance. 



The second overlooked aspect regards the lack of distributed optimization algorithms to solve convex relaxation problems with \emph{certificates of convergence} to the centralized optimizer, convergence rate, and proven robustness when applied to real sensor networks bounded by asynchronous communication and limited computation capabilities. 

\textbf{Contributions.} First, we generalize the current state-of-the-art convex relaxations by formulating the sensor network localization problem in a maximum likelihood framework and then relaxing it. This class of relaxations (which depends on the choice of the probability density function (PDF) of the noise) is represented by the convex program~\eqref{eq.sdp}. We show that this program is a rank relaxation of the original non-convex ML estimation problem, and at least for two widely used cases (Gaussian noise and Gaussian quantized measurements), it is a rank-$D$ relaxation ($D$ being the dimension of the space where the sensor nodes live, Proposition~\ref{prop.sdprelax}). The relaxed convex program is then further massaged into the edge-based ML relaxation~\eqref{eq.esdp} to lessen the computation requirements and to facilitate the distribution of the problem among the nodes. Furthermore, we show numerically that the tightness of the relaxation (in particular, the property of being derived from a rank-$D$ relaxation or not) can affect the performance of the convex program~\eqref{eq.esdp} more than the correctness of the noise model. 


As a second contribution, we demonstrate how the edge-based ML convex relaxation can be handled via the alternating direction method of multipliers (ADMM), which gives us a powerful leverage for the analysis of the resulting algorithm. The proposed algorithm, Algorithm~\ref{alg.ADMM}, is distributed in nature: the sensor nodes are able to locate themselves and the neighboring nodes without the knowledge of the whole network. This algorithm converges with a rate of $O(1/\kit)$ ($\kit$ being the number of iterations) to the solution of~\eqref{eq.esdp} (Theorem~\ref{theo.convergence}). Using Algorithm~\ref{alg.ADMM}, each sensor node has a total communication cost to reach a certain average local accuracy of the solution that is independent of the network size (Proposition~\ref{prop.time} and Corollary~\ref{coro.comm}). The proposed algorithm is then proven to converge even when run asynchronously (Theorem~\ref{theo.asynch}) and when the nodes are affected by computation errors (Theorem~\ref{prop.2}). These features, along with guaranteed convergence, are very important in real-life sensor network applications. Finally, we compare the usage of Algorithm~\ref{alg.ADMM} with some other available possibilities, in particular, the methods suggested in~\cite{Shi2010} and \cite{Simonetto2013a}, both in terms of theoretical performances and simulation results. These analyses support our proposed distributed algorithm, especially for large-scale settings.   


\textbf{Organization.} The remainder of the paper is organized as follows. Section~\ref{sec:pf} details the problem formulation. Section~\ref{sec:convex} presents the proposed maximum likelihood convex relaxation~\eqref{eq.sdp} along with some examples. Section~\ref{sec:edge} introduces the edge-based relaxation~\eqref{eq.esdp}, which is the building block for our distributed algorithm. Section~\ref{sec:distr0} surveys briefly distributed techniques to solve the localization problem, while, in Section~\ref{sec:distr}, we focus on the development of our distributed algorithm and its analysis. Numerical simulations and comparisons are displayed in Section~\ref{sec:num}, while our conclusions are drawn in Section~\ref{sec:concl}. 

\section{Preliminaries and Problem Statement}\label{sec:pf}

We consider a network of $n$ static wireless sensor nodes with computation and communication capabilities, living in a $D$-dimensional space (typically $D$ will be the standard 2-dimensional or 3-dimensional Euclidean space). We denote the set of all nodes $\mathcal{V} = \{1, \dots, n\}$. Let $\x_{i} \in \mathbb{R}^D$ be the position vector of the $i$-th sensor node, or equivalently, let $\X = [\x_{1} , \dots , \x_{n}] \in \mathbb{R}^{D\times n}$ be the matrix collecting the position vectors. We consider an environment with line-of-sight conditions between the nodes and we assume that some pairs of sensor nodes $(i,j)$ have access to noisy range measurements as	
\begin{equation}
r_{i,j} = d_{i,j} + \nu_{i,j},
\label{eq.noisymeasurement}
\end{equation}
where $d_{i,j} = ||\x_{i} - \x_{j}||_2$ is the noise-free Euclidean distance and $\nu_{i,j}$ is an additive noise term with known probability distribution. We call $p_{i,j}(d_{i,j}(\x_i,\x_j)|r_{i,j})$ the inter-sensor sensing PDF, where we have indicated explicitly the dependence of $d_{i,j}$ on the sensor node positions $(\x_i,\x_j)$. 

In addition, we consider that some sensors also have access to noisy range measurements with some fixed anchor nodes (whose position $\a_k$, for $k \in \{1,\dots,m\}$, is known by all the neighboring sensor nodes of each $\a_k$) as   
\begin{equation}
v_{i,k} = e_{i,k} + \mu_{i,k},
\label{eq.noisyanchormeasurement}
\end{equation}
where, $e_{i,k} = ||\x_{i} - \a_{k}||_2$ is the noise-free Euclidean distance and $\mu_{i,k}$ is an additive noise term with known probability distribution. We denote as $p_{i,k,\mathrm{a}}(e_{i,k}(\x_i,\a_k)|v_{i,k})$ the anchor-sensor sensing PDF.

We use graph theory terminology to characterize the set of sensor nodes $\mathcal{V}$ and the measurements $r_{i,j}$ and $v_{i,k}$. In particular, we say that the measurements $r_{i,j}$ induce a graph with $\mathcal{V}$ as vertex set, i.e., for each sensor node pair $(i,j)$ for which there exists a measurement $r_{i,j}$, there exists an edge connecting $i$ and $j$. The set of all edges is $\mathcal{E}$ and its cardinality is $E$. We denote this undirected graph as $\mathcal{G} = (\mathcal{V}, \mathcal{E})$. The neighbors of sensor node $i$ are the sensor nodes that are connected to $i$ with an edge. The set of these neighboring nodes is indicated with $\mathcal{N}_i$, that is $\mathcal{N}_i = \{j| (i, j) \in \mathcal{E}\}$. Since the sensor nodes are assumed to have communication capabilities, we implicitly assume that each sensor node $i$ can communicate with all the sensors in $\mathcal{N}_i$, and with these only. In a similar fashion, we collect the anchors in the vertex set $\mathcal{V}_\mathrm{a} = \{1, \dots, m\}$ and we say that the measurements $v_{i,k}$ induce an edge set $\mathcal{E}_\mathrm{a}$, composed by the pairs $(i,k)$ for which there exists a measurement $v_{i,k}$. Also, we denote with $\mathcal{N}_{i, \mathrm{a}}$ the neighboring anchors for sensor node $i$, i.e., $\mathcal{N}_{i, \mathrm{a}} = \{k| (i, k) \in \mathcal{E}_\mathrm{a}\}$.

\textbf{Problem Statement.} The sensor network localization problem is formulated as estimating the position matrix $\X$ (in some cases, up to an orthogonal transformation) given the measurements $r_{i,j}$ and $v_{i,k}$ for all $(i,j)\in \mathcal{E}$ and $(i,k) \in \mathcal{E}_{\mathrm{a}}$, and the anchor positions $\a_k$, $k \in \mathcal{V}_{\mathrm{a}}$. When $\mathcal{V}_{\mathrm{a}} = \emptyset$ we call the problem anchor-free localization. The sensor network localization problem can be written in terms of maximizing the likelihood leading to the following optimization problem
\begin{multline}
\hskip-0.3cm\X_{\textrm{ML}}^* = \arg\hskip-0.1cm\max_{\hskip-0.4cm \X\in\mathbb{R}^{D\times n}} \left\{\sum_{(i,j) \in\mathcal{E}} \ln p_{i,j}(d_{i,j}(\x_i,\x_j)|r_{i,j})\right.\\\left. +  \sum_{(i,k)\in\mathcal{E}_{\mathrm{a}}} \ln p_{i,k,\mathrm{a}}(e_{i,k}(\x_i,\a_k)|v_{i,k})\right\}.
\label{eq.mle}
\end{multline}
This optimization problem is in general non-convex and it is also NP-hard to find any global solution. 
In this paper, under the sole assumptions that:
\begin{assumption}\label{ass.sens}
The sensing PDFs $p_{i,j}(d_{i,j}(\x_i,\x_j)|r_{i,j})$ and $p_{i,k,\mathrm{a}}(e_{i,k}(\x_i,\a_k)|v_{i,k})$ are log-concave functions of the unknown distances $d_{i,j}$ and $e_{i,k}$,
\end{assumption}
\begin{assumption}\label{ass.conn}
The graph induced by the inter-sensor range measurements $\mathcal{G}$ is connected,
\end{assumption}
\vskip0.1cm
we will propose a convex relaxation to transform the ML estimator~\eqref{eq.mle} into a more tractable problem, which we will then solve using ADMM in a distributed setting, where each of the sensor nodes, by communicating only with the neighboring nodes, will determine its own position. 

\section{Convex Relaxations}\label{sec:convex}

\subsection{Maximum Likelihood Relaxation}


To derive the mentioned convex relaxation of the ML estimator~\eqref{eq.mle}, several steps are needed. First of all, we introduce the new variables $\Y = \X^\transp \X$, $\delta_{i,j} = d_{i,j}^2$, $\epsilon_{i,k} = e_{i,k}^2$, and we collect the $d_{i,j}$, $e_{i,k}$, $\delta_{i,j}$, $\epsilon_{i,k}$ scalar variables into the stacked vectors $\d, \e, \mathbold{\delta}, \mathbold{\epsilon}$. Second, we rewrite the cost function of the ML estimator as dependent only on the pair $(\d, \e)$ as 
\begin{multline}
f(\d, \e) := - \Big(\sum_{(i,j)\in\mathcal{E}}\ln p_{i,j}(d_{i,j}|r_{i,j}) \\  + \sum_{(i,k)\in\mathcal{E}_\mathrm{a}}  \ln p_{i,k,\mathrm{a}}(e_{i,k}|v_{i,k}) \Big).
\label{eq.mlede}
\end{multline}
Third, we re-introduce the dependencies of $(\d,\e)$ on $\X$ and on $(\mathbold{\delta}, \mathbold{\epsilon})$ by considering the following \emph{constrained} optimization 
\begin{subequations}\label{eq.mle1}
\begin{align}
\minimize_{\X, \Y, \mathbold{\delta}, \mathbold{\epsilon}, \d, \e} &&& f(\d,\e) & \label{costfuncde}\\
\textrm{subject to} &&& \hskip-0.25cm\left.\begin{array}{l}Y_{ii} + Y_{jj} - 2 Y_{ij} = \delta_{i,j}, \\ \delta_{i,j} = d_{i,j}^2, \,d_{i,j}\geq 0, \hskip0cm\textrm{for all } (i,j)\in\mathcal{E}\end{array} \right\}\label{eq.cons.lmi01}\\
&&& \hskip-0.25cm\left.\begin{array}{l} Y_{ii} - 2 \x_i^\transp \a_{k} + ||\a_k||^2_2 = \epsilon_{i,k}, \\ \epsilon_{i,k} = e_{i,k}^2,\, e_{i,k}\geq0, \hskip0cm \textrm{for all } (i,k)\in\mathcal{E}_{\mathrm{a}}\end{array} \right\}\label{eq.cons.lmi02}\\
&&& \Y = \X^\transp \X.  \label{eq.cons.lmi0X}
\end{align}
\end{subequations}
The problem~\eqref{eq.mle1} is \emph{equivalent} to~\eqref{eq.mle}: the constraints in the problem~\eqref{eq.mle1} have both the scope of imposing the pair-wise distance relations and of enforcing the chosen change of variables (in fact, without the constraints, all the variables would be independent of each other). In the new variables and under Assumption~\ref{ass.sens}, $f(\d,\e)$ is a convex function, however the constraints of~\eqref{eq.mle1} still define a non-convex set. Nonetheless, we can massage the constraints by using Schur complements and propose the following convex relaxation
\begin{subequations}\label{eq.sdp}
\begin{align}
\minimize_{\X, \Y, \mathbold{\delta}, \mathbold{\epsilon}, \d, \e} &&& f(\d,\e) & \label{eq.costfunction}\\
\textrm{subject to} &&& \hskip-0.4cm\left.\begin{array}{l} Y_{ii} + Y_{jj} - 2 Y_{ij} = \delta_{i,j}, \delta_{i,j} \geq 0,\\
\left(\begin{array}{cc} 1 & d_{i,j} \\ d_{i,j} & \delta_{i,j} \end{array}\right) \succeq 0,  d_{i,j}\geq 0, \textrm{for all }  (i,j)\in\mathcal{E}\end{array} \hskip-0.125cm\right\} \label{eq.cons.lmi1}\\
&&& \hskip-0.4cm\left.\begin{array}{l} Y_{ii} - 2 \x_i^\transp \a_{k} + ||\a_k||^2_2 = \epsilon_{i,k}, \epsilon_{i,k} \geq 0, \\  \left(\begin{array}{cc} 1 & e_{i,k} \\ e_{i,k} & \epsilon_{i,k} \end{array}\right) \succeq 0, e_{i,k}\geq 0, \textrm{for all }  (i,k)\in\mathcal{E}_{\mathrm{a}}\end{array}\hskip-0.3cm \right\}  \label{eq.cons.lmi2}\\
&&&\hskip-0.2cm \left(\begin{array}{cc} \mathbf{I}_D & \X \\ \X^\transp & \Y \end{array}\right) \succeq 0, \Y \succeq 0. \label{eq.lmiX}
\end{align}
\end{subequations}
The problem~\eqref{eq.sdp} is now convex (specifically, it is a convex optimization problem with generalized inequality constraints~\cite{Boyd2004a}) and its optimal solution represents a lower bound for the original non-convex ML estimator~\eqref{eq.mle}. 

In the problem~\eqref{eq.sdp}, all the three constraints~\eqref{eq.cons.lmi1} till \eqref{eq.lmiX} are rank relaxed versions of~\eqref{eq.cons.lmi01} till \eqref{eq.cons.lmi0X}, which makes problem~\eqref{eq.sdp} a rank relaxation. Usually, convex relaxations for sensor network localization are formulated directly on the squared distance variables $(\mathbold{\delta}, \mathbold{\epsilon})$ using a cost function $f_\mathrm{sq}(\mathbold{\delta}, \mathbold{\epsilon})$ (not ML) and eliminating the variables $(\d,\e)$. This way of formulating the problem does not capture the noise distribution, but renders the resulting relaxation a rank-$D$ relaxation, since~\eqref{eq.lmiX} is the only relaxed constraint~\cite{Biswas2006a}. Problem~\eqref{eq.sdp} both models correctly the noise distribution, being derived from an ML formulation, and for some common used noise PDFs can be transformed into a rank-$D$ relaxation, in which case it is equivalent in tightness to relaxations based on squared distance alone.  

In the next subsections, we specify the convex relaxation~\eqref{eq.sdp} for different noise distributions (satisfying Assumption~\ref{ass.sens}) and prove that~\eqref{eq.sdp} can be expressed as a rank-$D$ relaxation for two particular yet widely used cases.  In Section~\ref{sec:num}, while presenting simulation results, we discuss how this aspect can affect the quality of the position estimation. In particular, it appears that tighter relaxations may have a lower estimation error, even when they employ less accurate noise models. 

\subsection{Example 1-- Gaussian Noise Relaxation}

\begin{figure*}[bp]\footnotesize
\hrule
\begin{multline*}
f_{\mathrm{Q, GN}}(\mathbold{\delta}, \mathbold{\epsilon},\d,\e):= 
 -\left(\sum_{(i,j)\in\mathcal{E}} \ln \left( \sum_{s=1}^S q_{r,i,j,s} \int_{r'_{i,j}\in\mathcal{Q}_s} \hskip-.75cm\exp\left[-\sigma^{-2}_{i,j}/2(\delta_{i,j} -2 d_{i,j}r_{i,j}'+r_{i,j}'^2)\right]\textrm{d}r'_{i,j}\right) +\right.\\ \left. \sum_{(i,k)\in\mathcal{E}_{\mathrm{a}}} \hskip-0.2cm\ln \left(\sum_{s=1}^S q_{v,i,k,s} \int_{v'_{i,k}\in\mathcal{Q}_s} \hskip-0.75cm\exp\left[-\sigma^{-2}_{i,k,\mathrm{a}}/2(\epsilon_{i,k} -2 e_{i,k}v_{i,k}'+v_{i,k}'^2)\right]\textrm{d}v'_{i,k}\right)\right)
\end{multline*}
\end{figure*}

In the case of Gaussian noise, we assume that the noises $\nu_{i,j}$ and $\mu_{i,k}$ in the sensing equations~\eqref{eq.noisymeasurement} and \eqref{eq.noisyanchormeasurement} are drawn from a white zero-mean PDF, i.e., $\nu_{i,j}\sim \mathcal{N}(0, \sigma_{i,j}^2)$ and $\mu_{i,k}\sim \mathcal{N}(0, \sigma_{i,k,\mathrm{a}}^2)$. The cost function $f(\d,\e)$ then is
\begin{multline*}
f_{\mathrm{GN},0}(\d,\e):=\sum_{(i,j)\in\mathcal{E}} \sigma_{i,j}^{-2}(d_{i,j}^2-2 d_{i,j} r_{i,j} + r_{i,j}^2) + \\ \sum_{(i,k)\in\mathcal{E}_{\mathrm{a}}} \sigma_{i,k,\mathrm{a}}^{-2}(e_{i,k}^2-2 e_{i,k} v_{i,k} + v_{i,k}^2).
\end{multline*}
A natural way to rewrite this cost is to enforce the change of variables $\delta_{i,j}=d_{i,j}^2$ and $\epsilon_{i,k} = e_{i,k}^2$, yielding 
\begin{multline*}
f_{\mathrm{GN}}(\mathbold{\delta},\mathbold{\epsilon},\d,\e):=\sum_{(i,j)\in\mathcal{E}} \sigma_{i,j}^{-2}(\delta_{i,j}-2 d_{i,j} r_{i,j} + r_{i,j}^2) + \\ \sum_{(i,k)\in\mathcal{E}_{\mathrm{a}}} \sigma_{i,k,\mathrm{a}}^{-2}(\epsilon_{i,k}-2 e_{i,k} v_{i,k} + v_{i,j}^2).
\end{multline*}
With the cost $f_{\mathrm{GN}}(\mathbold{\delta},\mathbold{\epsilon},\d,\e)$, the optimization problem reads\footnote{A similar formulation for this relaxation can be found in~\cite{Biswas2006a}. We note that problem~\eqref{eq.sdpGN} is not equivalent to~\eqref{eq.sdp} with cost function $f_{\mathrm{GN,0}}(\d,\e)$, since for~\eqref{eq.sdp}, $\delta_{i,j}\geq d_{i,j}^2$ and $\epsilon_{i,k}\geq e_{i,k}^2$. }
\begin{subequations}\label{eq.sdpGN}
\begin{align}
\minimize_{\X, \Y, \mathbold{\delta}, \mathbold{\epsilon}, \d, \e} &&& f_{\mathrm{GN}}(\mathbold{\delta}, \mathbold{\epsilon},\d,\e) & \label{eq.costfunctionGN}\\
\textrm{subject to} &&& \eqref{eq.cons.lmi1}, \eqref{eq.cons.lmi2}, \eqref{eq.lmiX}
\end{align}
\end{subequations}
This relaxation is not only convex but also a semidefinite program (SDP), i.e., it has a linear cost function and generalized linear constraints~\cite{Boyd2004a}. Some of its constraints are linear matrix inequalities (LMIs). For the semidefinite program~\eqref{eq.sdpGN}, the following proposition holds true. 
\begin{proposition}\label{prop.sdprelax}
Under the assumption of Gaussian noise, the semidefinite program~\eqref{eq.sdpGN} is a rank-$D$ relaxation of the original non-convex optimization problem~\eqref{eq.mle}.
\end{proposition}
\begin{proof}
We need to show that at optimality the relaxed constraints~\eqref{eq.cons.lmi1} and \eqref{eq.cons.lmi2} are equivalent to the original constraints~\eqref{eq.cons.lmi01} and \eqref{eq.cons.lmi02}. In other words, we need to show that any optimal solution of the semidefinite program~\eqref{eq.sdpGN}, say $\mathbold{\delta}^*, \mathbold{\epsilon}^*, \d^*, \e^*, \Y^*, \X^*$, satisfies the following
$$
\delta_{i,j}^* = d_{i,j}^{2*}, \textrm{ and } \epsilon_{i,k}^* = e_{i,k}^{2*}
$$ 
for all $(i,j) \in \mathcal{E}$ and for all $(i,k) \in \mathcal{E}_{\mathrm{a}}$. To see this, note that the LMIs in the constraints~\eqref{eq.cons.lmi1} and \eqref{eq.cons.lmi2} can be rewritten as
\begin{equation}\label{eq.sole}
d_{i,j}^2 \leq \delta_{i,j},\textrm{and } e_{i,k}^2 \leq \epsilon_{i,k}.
\end{equation}
The cost function~\eqref{eq.costfunctionGN} maximizes the scalar variables $d_{i,j}$ and $e_{i,k}$, which are constrained only by~\eqref{eq.sole}. Therefore at optimality, we will always have $d_{i,j}^{2*} = \delta_{i,j}^*$ and $e_{i,k}^{2*} = \epsilon_{i,k}^*$, and thus the claim holds.
\end{proof}






\subsection{Example 2-- Quantized Observation Relaxation}

An interesting, and realistic, elaboration of the ML estimator is when, due to limited sensing capabilities, the sensors produce a \emph{quantized} version of $r_{i,j}$ and $v_{i,k}$ (see the discussion in ~\cite{Schizas2008,Jakubiec2013} for its relevance in sensor networks). Consider an $S$-element convex tessellation of $\mathbb{R}_{+}$, comprised of the convex sets $\{\mathcal{Q}_{s}\}_{s=1}^{S}$. A quantization of $r_{i,j}$ and $v_{i,k}$ produces the observations $q_{r,i,j,s}$ and $q_{v,i,k,s}$, which are unitary if $r_{i,j}\in \mathcal{Q}_{s}$ and $v_{i,k}\in\mathcal{Q}_{s}$, respectively. Otherwise $q_{r,i,j,s}$ and $q_{v,i,k,s}$ are zero. The resulting cost function for the convex relaxation~\eqref{eq.sdp} is now $f_{\mathrm{Q}}(\d,\e):=$ \vskip-0.5cm
\begin{multline*}
-\left(\sum_{(i,j)\in\mathcal{E}} \ln \left(\sum_{s=1}^S q_{r,i,j,s}\int_{r'_{i,j}\in\mathcal{Q}_s} p_{i,j}(d_{i,j}|r'_{i,j})\textrm{d}r'_{i,j} \right) + \right.\\ 
\left.\sum_{(i,k)\in\mathcal{E}_{\mathrm{a}}} \ln \left(\sum_{s=1}^S q_{v,i,k,s}\int_{v'_{i,k}\in\mathcal{Q}_s} p_{i,k,\mathrm{a}}(e_{i,k}|v'_{i,k})\textrm{d}v'_{i,k} \right)\right),
\end{multline*}
which is convex, since the integral of a log-concave function over a convex set is also log-concave. The resulting convex relaxation reads
\begin{subequations}\label{eq.sdpQ}
\begin{align}
\minimize_{\X, \Y, \mathbold{\delta}, \mathbold{\epsilon}, \d, \e} &&& f_{\mathrm{Q}}(\d,\e) & \label{eq.costfunctionQ}\\
\textrm{subject to} &&& \eqref{eq.cons.lmi1}, \eqref{eq.cons.lmi2}, \eqref{eq.lmiX}
\end{align}
\end{subequations}
which is a rank relaxation of~\eqref{eq.mle}, but in general not a rank-$D$ relaxation. We can specify~\eqref{eq.costfunctionQ} for Gaussian noise (using the same variable enforcing of $f_{\mathrm{GN}}$) as done in the equation at the bottom of the page. 
It is not difficult to show that the convex relaxation~\eqref{eq.sdpQ} equipped with the cost $f_{\mathrm{Q, GN}}(\mathbold{\delta}, \mathbold{\epsilon},\d,\e)$ is now a rank-$D$ relaxation, by using similar arguments as in Proposition~\ref{prop.sdprelax}. 


\subsection{Example 3-- Laplacian Noise Relaxation}

Laplacian noise is used for example to model outliers in range measurements~\cite{Ouguz-Ekim2011} and to model errors coming from signal interference, e.g., in UWB localization systems~\cite{Wymeersch2012}. In the Laplacian noise case the cost function can be specified as 
\begin{multline*}
f_{\mathrm{L}}(\d,\e):=\sum_{(i,j)\in\mathcal{E}} \frac{|d_{i,j} - r_{i,j}|}{\sigma_{i,j}} + \sum_{(i,k)\in\mathcal{E}_{\mathrm{a}}} \frac{|e_{i,k} - v_{i,k}|}{\sigma_{i,k,\mathrm{a}}},
\end{multline*}
and the ML convex relaxation reads
\begin{subequations}\label{eq.esdpunif}
\begin{align}
\minimize_{\X, \Y, \mathbold{\delta}, \mathbold{\epsilon}, \d, \e} &&& f_{\mathrm{L}}(\d,\e) \\
\textrm{subject to} &&& \eqref{eq.cons.lmi1}, \eqref{eq.cons.lmi2}, \eqref{eq.lmiX}. 
\end{align}
\end{subequations}

This ML convex relaxation is neither a rank-$D$ relaxation, nor it can be transformed into one by some variable enforcing in the cost function, yet it correctly models Laplacian PDFs.

\subsection{Example 4- Uniform Noise Relaxation}

Uniform noise distributions are used when the source of error is not known a priori and only a bound on the noise level is available. For example, this is the case when we are aware of a lower bound on the pair-wise distances and of an upper bound dictated by connectivity \cite{Biswas2008, Sheu2010}. Considering uniform noise PDFs in the range $d_{i,j}\pm \sigma_{i,j}$ and $e_{i,k}\pm\sigma_{i,k,\mathrm{a}}$,    
the convex relaxation~\eqref{eq.sdp} becomes the following feasibility problem
\begin{subequations}\label{eq.esdpunif}
\begin{align}
\mathrm{find}&&&{\X, \Y, \mathbold{\delta}, \mathbold{\epsilon}, \d, \e} \\
\textrm{such that} &&& \eqref{eq.cons.lmi1}, \eqref{eq.cons.lmi2}, \eqref{eq.lmiX}  \\
&&& \hskip-2cm r_{i,j}-\sigma_{i,j}\leq d_{i,j}\leq r_{i,j}+\sigma_{i,j} \textrm{ for all} (i,j)\in\mathcal{E}\\
&&& \hskip-2cm v_{i,k}-\sigma_{i,k,\mathrm{a}}\leq e_{i,k}\leq v_{i,k}+\sigma_{i,k,\mathrm{a}} \textrm{ for all} (i,k)\in\mathcal{E}_{\mathrm{a}}.
\end{align}
\end{subequations}

Also in this case, the ML convex relaxation is neither a rank-$D$ relaxation, nor it can be transformed into one by some variable enforcing in the cost function, yet it correctly models uniform noise distributions.

\section{Edge-based Convex Relaxations}\label{sec:edge}

The convex relaxations derived from~\eqref{eq.sdp} couple arbitrarily far away sensor nodes through the LMI constraint~\eqref{eq.lmiX}. This complicates the design of a distributed optimization algorithm. In addition, due to ~\eqref{eq.lmiX}, the complexity of solving the semidefinite program~\eqref{eq.sdp} scales at least as $O(n^3)$, i.e., is at least cubic in the number of sensor nodes \cite{Boyd2004a}, and it could become unfeasible for large-scale networks. In order to massage this coupling constraint, we introduce a further relaxation for~\eqref{eq.sdp}, which will be called edge-based ML (E-ML) relaxation. We consider the following relaxation of~\eqref{eq.sdp} 

\begin{subequations}\label{eq.esdp}
\begin{align}
\minimize_{\X, \Y, \mathbold{\delta}, \mathbold{\epsilon}, \d, \e} &&& f(\d,\e)  \label{eq.costfunctionESDP}\\
\textrm{subject to} &&& \eqref{eq.cons.lmi1}, \eqref{eq.cons.lmi2} \label{eq.cons.10b}\\
&&& \hskip-0.2cm\left(\begin{array}{c|cc} \mathbf{I}_D & \x_i & \x_j  \\ \hline  \x_i^\transp & Y_{ii} &Y_{ij} \\ 
\x_j^\transp & Y_{ij}& Y_{jj} \end{array}\right) \succeq 0,\textrm{for all } (i,j) \in \mathcal{E}. \label{c.lmi3}
\end{align}
\end{subequations}

This relaxation employs the same idea of the edge-based semidefinite program (ESDP) relaxation of~\cite{Wang2008, Lui2009} of considering the coupling constraint~\eqref{eq.lmiX} to be valid on the edges only. Since the constraint~\eqref{eq.lmiX} implies \eqref{c.lmi3} but not the contrary, the relaxation~\eqref{eq.esdp} is not a rank-$D$ relaxation. However, it is straightforward to see that, if the original convex relaxation~\eqref{eq.sdp} was a rank-$D$ relaxation, then for the derived \eqref{eq.esdp}, it would be true that $\delta_{i,j}^* = d_{i,j}^{2*}, \epsilon_{i,k}^* = e_{i,k}^{2*}$. For example, this is the case for Gaussian noise, and we show how this can play an important role for the accuracy in Section~\ref{sec:num}. 


The convex relaxation~\eqref{eq.esdp} is now ready to be distributed among the sensor nodes.


\section{Distributed Algorithms for Sensor Network Localization}\label{sec:distr0}

Different distributed methods for sensor network localization have been proposed in recent years. A first group consists of heuristic algorithms, which are typically based on the paradigm of dividing the nodes into arbitrarily selected clusters, solving the localization problem within every cluster and then patching together the different solutions. Methods that belong to this group are  \cite{Chan2009, Khan2010,Cucuringu2012}, while heuristic approaches to SDP relaxations are discussed in~\cite{Biswas2008}. Among the disadvantages of the heuristic approaches is that we introduce arbitrariness into the problem and we typically lose all the guarantees of performance of the ``father'' centralized approach. Furthermore, very often these heuristic methods are \emph{ad-hoc} and problem-dependent, which makes their theoretical characterization difficult (in contrast with the usage of well-established decomposition methods~\cite{Bertsekas1997}). 

The second group of methods employs decomposition techniques to guarantee that the distributed scheme converges to the centralized formulation asymptotically. In this group, under the Gaussian noise assumption, we can find methods that tackle directly the non-convex optimization problem~\eqref{eq.mle} with parallel gradient-descent iterative schemes~\cite{Costa2005,Calafiore2012} or (very recently) a work that uses a minimization-majorization technique to massage~\eqref{eq.mle} sequentially and then employs the alternating direction method of multipliers (ADMM) to distribute the computations among the sensor nodes~\cite{Soares2012}. These approaches have certificates of convergence to a local minimum of the original non-convex problem\footnote{This may not be sufficient for a reasonable localization; thus the need for a good starting condition which can be provided by convex relaxations, see~\cite{Weinberger2007} for some interesting numerical examples.}. Other methods encompass algorithms that tackle multi-dimensional scaling with a communication-intensive distributed spectral decomposition~\cite{Montanari2010}, and algorithms that tackle instead the convex SOCP/SDP relaxations~\cite{Srirangarajan2008,Shi2010,Simonetto2013a}. In particular~\cite{Srirangarajan2008} proposes a parallel distributed version of an SOCP relaxation (similar to the ESDP in~\cite{Wang2008}), whose convergence properties are however not analyzed\footnote{As a matter of fact, the proposed Jacobi-like algorithm is very hard to be proven converging to the centralized solution, since the constraints are coupled and not Cartesian, see~\cite{Bertsekas1997} for a detailed discussion.}. In~\cite{Shi2010}, the authors propose a further improvement of~\cite{Srirangarajan2008} based on the Gauss-Seidel algorithm, which is sequential in nature (meaning that sensors have to wait for each other before running their own local algorithm) and offers convergence guarantees to the ESDP of~\cite{Wang2008}. However, due to the sequential nature, the convergence \emph{rate} depends on the number of sensor nodes, which makes the approach impractical for large-scale networks. Finally, in \cite{Simonetto2013a} duality is exploited to design inexact primal-dual iterative algorithms based on the convex relaxation of~\cite{Weinberger2006, Sun2006, Weinberger2007}. This last approach has the advantage to be parallel and not sequential, nonetheless it is based on consensus algorithms whose convergence rate is also dependent on the size of the network, thus less practical for a large number of sensor nodes.  

In the next section, we propose a distributed algorithm based on ADMM to solve the edge-based convex relaxation~\eqref{eq.esdp}. The algorithm is proven to converge to the centralized optimizer as $O(1/\kit)$, where $\kit$ is the number of iterations. Furthermore, the computation and communication per iteration and per node do not depend on the size of the network, but only on the size of each one's neighborhood. Finally, we prove that the algorithm converges also in the case of asynchronous communication protocols and computation errors, making it robust to these two common issues in sensor networks.


\section{Proposed Distributed Approach}\label{sec:distr}

\subsection{Preliminaries and Background on ADMM}


In order to present our distributed algorithm, first of all, we rewrite the convex program~\eqref{eq.esdp} in a more compact way. Define the \emph{shared} vector
$$
\z_{i,j} := (Y_{ii}, Y_{jj}, Y_{ij}, \delta_{i,j}, d_{i,j}, \x_{i}^\transp, \x_{j}^\transp)^\transp \;\in\mathbb{R}^{5+2D},
$$
for each $(i,j) \in \mathcal{E}$ and call $\z$ the stacked vector comprised of all the $\z_{i,j}$'s. In a similar fashion, define the \emph{local} vector   
$$
\p_{i} := (\mathbold{\epsilon}_{i}^\transp, \e_{i}^\transp, Y_{ii}, \x_{i}^\transp)^\transp\;\in\mathbb{R}^{2|\mathcal{N}_{i,\mathrm{a}}|+D+1},
$$
where $\mathbold{\epsilon}_{i}$ and $\e_{i}$ are the concatenated vectors of $\epsilon_{i,k}$ and $e_{i,k}$ for all $k \in \mathcal{N}_{\mathrm{a},i}$, and call $\p$ the stacked vector of all the $\p_i$'s. We note that $\p_i$ and $\z_{i,j}$ are not independent, but this will not be an issue. Moreover, define the convex sets 
$$
\mathcal{Z}_{i,j} :=\{\z_{i,j}|\z_{i,j} \textrm{ verifies }\eqref{eq.cons.lmi1} \textrm{ and }\eqref{c.lmi3} \},
$$
$$
\mathcal{P}_{i} :=\{\p_{i}|\p_{i} \textrm{ verifies }\eqref{eq.cons.lmi2} \}.
$$
Problem~\eqref{eq.esdp} is then equivalent to
\begin{subequations}\label{eq.esdpsimplified}
\begin{align}
\minimize_{\z, \p} &&& f(\z, \p) & \label{eq.costfunction}\\
\textrm{subject to} &&& \z_{i,j} \in \mathcal{Z}_{i,j} &&  \textrm{for all } (i,j) \in \mathcal{E} \\
&&& \p_{i} \in \mathcal{P}_{i}  \hskip0.75cm  &&  \textrm{for all } i \in \mathcal{V}
\end{align}
\end{subequations}
where, for the general case,
\begin{multline}\label{cost}
f(\z, \p) := -\Big(\sum_{(i,j)\in\mathcal{E}} \ln p_{i,j}(d_{i,j}|r_{i,j}) +\\ \sum_{(i,k)\in\mathcal{E}_{\mathrm{a}}}\ln p_{i,k,\mathrm{a}}(e_{i,k}|v_{i,k})\Big) =: \sum_{i\in\mathcal{V}}  f_i(\z, \p_i).
\end{multline}

From the structure of the cost~\eqref{cost} and the problem~\eqref{eq.esdpsimplified} one can already see that the convex optimization~\eqref{eq.esdpsimplified} is separable and has $\z_{i,j}$ as complicating variables. One possible way to handle this type of optimization problems in a distributed way is employing the alternating direction  method of multipliers (ADMM). The reader is referred to~\cite{Boyd2011} for a very recent survey of this rather old technique, to the papers~\cite{Schizas2008, Zhu2009} which span possible applications of the method in signal processing, and to the mentioned recent work~\cite{Soares2012} that employs ADMM for a localization problem (albeit with a different flavor as the one presented here and applied to a different Gaussian noise-based approximated version of the original non-convex problem). In a nutshell, the strategy of ADMM is to assign copies of the coupling variable $\z_{i,j}$ to both node $i$ and node $j$ and then constrain these copies to be equal. The strength of ADMM, and the main reason of its employment in this paper, resides in its noise-resilience and computation error-resilience as well as the very loose assumptions required to guarantee its convergence (in contrast with typical dual, or primal-dual decomposition schemes.) 

In order to apply ADMM to the problem~\eqref{eq.esdpsimplified}, we define the local versions of the vector $\z_{i,j}$ as $\yiji$ and $\yijj$, meaning that $\yiji$ represents the vector $\z_{i,j}$ as seen by the node $i$, while $\yijj$ represents the vector $\z_{i,j}$ as seen by the node $j$. Call now the stacked vectors $\y_i^i$ as the ones comprised of $\yiji$ for all the $j\in\mathcal{N}_i$. We can then rewrite~\eqref{eq.esdpsimplified} in yet another \emph{equivalent} form as 
\begin{subequations}\label{eq.esdpsimplifiedadmm}
\begin{align}
\minimize_{\y_{1}^1,\dots, \y_{n}^n, \p, \z} &&& \sum_{i\in\mathcal{V}} f_i(\y_i^i, \p_i) & \label{eq.costfunction}\\
\textrm{subject to\qquad\,} &&& \hskip-1cm\yiji \in \mathcal{Z}_{i,j}, \yijj \in \mathcal{Z}_{i,j} &&\textrm{for all } (i,j) \in \mathcal{E} \\
&&& \hskip-1cm\p_{i} \in \mathcal{P}_{i} \hskip0.75cm  &&\textrm{for all } i \in \mathcal{V} \\
&&& \hskip-1.275cm\left.\begin{array}{l} \yiji - \z_{i,j} = 0\\ \yijj - \z_{i,j} = 0\end{array}\right\} &&\textrm{for all } (i,j) \in \mathcal{E}\label{coupleeq}
\end{align}
\end{subequations}

Problems~\eqref{eq.esdp}, \eqref{eq.esdpsimplified}, and \eqref{eq.esdpsimplifiedadmm} are all equivalent, but problem~\eqref{eq.esdpsimplifiedadmm} is better suited for ADMM, as we are about to see.


\begin{remark}
We remark that the sequential greedy optimization (SGO) method of~\cite{Shi2010} can also be applied to~\eqref{eq.esdpsimplified}. However, its analytical properties, such as convergence rate, noise-resilience, and computation error-resilience are still unknown at the moment; furthermore, its convergence has been proven only under the strong assumption of decoupled constraints (and argued for the real case of sparse coupling constraints, see~\cite{Shi2010}, Remark 4). Nonetheless, we will implement a distributed algorithm using SGO applied to our E-ML formulation to compare its performance analytically and numerically to ADMM. We will argue that SGO, given its sequential nature, is less suitable for large-scale networks. 
\end{remark}




\subsection{Proposed Algorithm}

The first step to derive the ADMM algorithm is, given a scalar $\rho>0$, defining the regularized Lagrangian of problem~\eqref{eq.esdpsimplifiedadmm} as
\begin{multline}\label{reglagr}
\La(\y,\p, \z,\la) := \\ 
\sum_{i\in\mathcal{V}} f_i(\y_i^i,\p_i) + \sum_{(i,j)\in\mathcal{E}}\left[\la_{i,j}^{i\transp}(\yiji-\z_{i,j}) + \la_{i,j}^{j\transp}(\yijj-\z_{i,j})\right] \\+ \sum_{(i,j)\in\mathcal{E}} \frac{\rho}{2}\left(\left\|\yiji-\z_{i,j} \right\|_2^2+\left\|\yijj-\z_{i,j} \right\|_2^2 \right) 
\end{multline}
where $\y$ is the shorthand notation for the vector $(\y_{1}^{1\transp},\dots, \y_{n}^{n\transp})^\transp$, while $\la$ is the shorthand notation for the vector of multipliers. To each couple of equality constraints~\eqref{coupleeq} we assign the multipliers $\la_{i,j}^i$ and $\la_{i,j}^j$. 

Solving~\eqref{eq.esdpsimplifiedadmm} with ADMM means implementing the following recursion: initialize the variables $\y^{(0)},\p^{(0)}, \z^{(0)},\la^{(0)}$, then
\begin{subequations}\label{admm0}
\begin{align}
(\y^{(\kit+1)},&\p^{(\kit+1)}) = \argmin_{\y\in\mathcal{Z}, \p\in\mathcal{P}} \{\La(\y,\p, \z^{(\kit)},\la^{(\kit)})\}, \\
\z^{(\kit+1)} &= \argmin_{\z} \{\La(\y^{(\kit+1)},\p^{(\kit+1)}, \z,\la^{(\kit)}) \}, \\
\la^{(\kit+1)} &= \la^{(\kit)} + \rho \nabla_\la\left[\La(\y^{(\kit+1)},\p^{(\kit+1)}, \z^{(\kit+1)},\la)\right]_{\la = \la^{(\kit)}} ,
\end{align}
\end{subequations}
for all $\kit\geq 0$, and with the convex sets $\mathcal{Z}$ and $\mathcal{P}$ defined as the union of the sets $\mathcal{Z}_{ij}$ for all $(i,j)\in\mathcal{E}$ and $\mathcal{P}_i$ for all $i\in\mathcal{V}$, respectively. 

In our sensor network application, the recursion~\eqref{admm0} is distributed in nature, since it can be carried out as follows. 

\begin{subequations}\label{admm1}
\begin{enumerate}
\item Set $\y_{i,j}^{i\,(0)}$, $\p_i^{(0)}$, $\z_{i,j}^{(0)}$, $\la_{i,j}^{i\,(0)}$, $\la_{i,j}^{j\,(0)}$ to zero, for all the nodes. 
\item At each iteration $\kit$, each node owns the variables $\y_{i,j}^{i\,(\kit)}$, $\p_i^{(\kit)}$, $\z_{i,j}^{(\kit)}$, $\la_{i,j}^{i\,(\kit)}$, $\la_{i,j}^{j\,(\kit)}$ for all $j\in\mathcal{N}_i$;
\item Each node updates its local variables $\y_{i,j}^{i(\kit)}$, $\p_i^{(\kit)}$ as
\begin{multline}\label{eq.argmin1}
\hskip-0.3cm(\y_{i}^{i\,(\kit+1)}, \p_i^{(\kit+1)}) = \arg\min_{\yiji \in \mathcal{Z}_{i,j}, \p_i\in\mathcal{P}_i} \left\{ f_i(\y_i^i,\p_i) + \right. \\ \left. \sum_{j\in\mathcal{N}_i} \left[\la_{i,j}^{i\,(\kit)\transp}\yiji + \frac{\rho}{2}\left\|\yiji-\z_{i,j}^{(\kit)} \right\|_2^2 \right] \right\}\,;
\end{multline}
\item Each node sends its local vector $\y_{i,j}^{i\,(\kit+1)}$ to its neighbor $j$, for all $j\in\mathcal{N}_i$;
\item Each node computes, for all $j\in\mathcal{N}_i$ 
\begin{multline}\label{eq.zit}
\z_{i,j}^{(\kit+1)} = \arg\min_{\z_{i,j}} \left\{ -\left(\la_{i,j}^{i\,(\kit)\transp} + \la_{i,j}^{j\,(\kit)\transp}\right)\z_{i,j}+\right.\\\left. \frac{\rho}{2}\left(\left\|\y_{i,j}^{i\,(\kit+1)}-\z_{i,j} \right\|_2^2+\left\|\y_{i,j}^{j\,(\kit+1)}-\z_{i,j} \right\|_2^2 \right) \right\}.
\end{multline}
We note here that since all the vectors $\y_{i,j}^{i\,(\kit+1)}$ are transmitted perfectly, the value of $\z_{i,j}^{(\kit+1)}$ computed by node $i$ is the same as the one computed by node $j$;
\item Each node computes, for all $j\in\mathcal{N}_i$
\begin{equation}\label{la.it2}
\begin{array}{ccc}\la_{i,j}^{i\,(\kit+1)} &=& \la_{i,j}^{i\,(\kit)} + \rho (\y_{i,j}^{i\,(\kit+1)} - \z_{i,j}^{(\kit+1)}) \\
\la_{i,j}^{j\,(\kit+1)} &=& \la_{i,j}^{j\,(\kit)} + \rho (\y_{i,j}^{j\,(\kit+1)} - \z_{i,j}^{(\kit+1)}).\end{array}
\end{equation}
We note that here also the values of $\la_{i,j}^{i\,(\kit+1)}$ and $\la_{i,j}^{j\,(\kit+1)}$ computed by node $i$ are the same as the ones computed by node $j$;
\item Set $\kit \leftarrow \kit+1$ and go to 2).
\end{enumerate}
\end{subequations}

We note that both the optimization problems~\eqref{eq.argmin1} and \eqref{eq.zit} are convex programs. Problem~\eqref{eq.argmin1} is an SDP (which can be solved using standard convex optimization toolboxes, such as Yalmip or CVX). In order to see this more clearly, Program~\eqref{eq.argmin1} needs to be written in the equivalent form
\begin{subequations}\label{eq.problemfinal}
\begin{align}
\minimize_{\yiji \in \mathcal{Z}_{i,j}, \p_i\in\mathcal{P}_i, \mathbold{\gamma}} &&&  f_i(\y_i^i,\p_i) + \sum_{j\in\mathcal{N}_i} \left[\la_{i,j}^{i\,(\kit)\transp}\yiji + \frac{\rho}{2} \gamma_{i,j} \right] \\
\textrm{subject to \color{white}aa\color{black} } &&& \hskip-.2cm\left(\begin{array}{cc} 1 & (\yiji-\z_{i,j}^{(\kit)} )^\transp \\ (\yiji-\z_{i,j}^{(\kit)} ) & \gamma_{i,j}\mathbf{I}_{5+2D} \end{array}\right) \succeq 0,\label{eq.gamma1}\\&&& \gamma_{i,j}\geq 0 \hskip2cm  \textrm{for all } (i,j) \in \mathcal{N}_i,\label{eq.gamma2}
\end{align}
\end{subequations}
where each $\gamma_{i,j}$ and the vector containing them $\mathbold{\gamma}$ are slack variables used to impose the quadratic penalty.

Problem~\eqref{eq.zit} is an unconstrained quadratic program in $\z_{i,j}$, whose solution is 
\begin{equation}\label{eq.z.ita}
\z_{i,j}^{(\kit+1)} = \frac{1}{2}\left[\y_{i,j}^{i\,(\kit+1)}+\y_{i,j}^{j\,(\kit+1)} + \frac{1}{\rho}\left(\la_{i,j}^{i\,(\kit)} + \la_{i,j}^{j\,(\kit)}\right)\right].
\end{equation}

We can now simplify the relations~\eqref{eq.z.ita} and \eqref{la.it2}. By using the relations~\eqref{la.it2}, we can write~\eqref{eq.z.ita} as
\begin{multline*}
\z_{i,j}^{(\kit+1)} = \frac{1}{2}\left[\y_{i,j}^{i\,(\kit+1)}+\y_{i,j}^{j\,(\kit+1)}\right] +\\ \frac{1}{\rho}\left(\la_{i,j}^{i\,(\kit-1)} + \la_{i,j}^{j\,(\kit-1)} + \rho(\y_{i,j}^{i\,(\kit)}+\y_{i,j}^{j\,(\kit)}-2 \z_{i,j}^{(\kit)})\right),
\end{multline*}
and, using again~\eqref{eq.z.ita} for $\z_{i,j}^{(\kit)}$, we obtain 
\begin{equation}\label{eq.z.itf}
\z_{i,j}^{(\kit+1)} = \frac{1}{2}\left[\y_{i,j}^{i\,(\kit+1)}+\y_{i,j}^{j\,(\kit+1)}\right].
\end{equation}

Furthermore, the following relations hold as by-products\footnote{Recall that we have set $\la_{i,j}^{i\,(0)} = \mathbf{0}$ and $\la_{i,j}^{j\,(0)} = \mathbf{0}$ and apply relations~\eqref{la.it2} and \eqref{eq.z.itf} recursively.} of~\eqref{eq.z.itf} for all $\kit \geq 0$:
\begin{equation*}
\la_{i,j}^{i\,(\kit+1)} + \la_{i,j}^{j\,(\kit+1)} = 0, \la_{i,j}^{i\,(\kit+1)} = \sum_{\kappa=1}^{\kit+1} \frac{\rho}{2}\left[\y_{i,j}^{i\,(\kappa)}-\y_{i,j}^{j\,(\kappa)}\right].
\end{equation*}

This simplifies the ADMM algorithm defined by the iterations~\eqref{admm1} (as summarized in Algorithm~\ref{alg.ADMM}), in particular the computation of $\la_{i,j}^{j\,(\kit+1)}$ is no longer required (as not needed any more for the computation of $\z_{i,j}$). 

\subsection{Properties of Algorithm~\ref{alg.ADMM} (Ideal)}

We now analyze the analytical properties of Algorithm~\ref{alg.ADMM} in terms of convergence and convergence rate of its solution to the optimal solution of the centralized problem~\eqref{eq.esdp}. As a by-product, we also characterize the number of iterations required to reach a given accuracy and the total communication cost. 

Let $\q$ denote the stacked vector of the optimization variables, i.e., $\q:=({\y}^\transp, {\p}^\transp, \z^\transp)^\transp$, and let $\bar{\q}_\kit$ represent the running averages, i.e., 
\begin{equation}\label{eq.runningaverages} 
\bar{\q}_\kit = \frac{1}{\kit+1}\sum_{\kappa = 0}^\kit {\q}^{(\kappa)},
\end{equation}
with ${\q}^{(\kappa)} = ({\y}^{(\kappa)\transp}\hskip-0.1cm, {\p}^{(\kappa)\transp}\hskip-0.1cm, \z^{(\kappa)\transp})^\transp$. Assume that the initial convex problem~\eqref{eq.esdpsimplifiedadmm} admits a solution and let $(\q^*, \la^*)$ be this solution. Then the following convergence theorem holds. 

\begin{theorem}~\label{theo.convergence}
Let $\q^{(\kit)}$ be the solution generated with Algorithm~\ref{alg.ADMM} applied to the convex problem~\eqref{eq.esdpsimplifiedadmm}, whose solution is denoted by $(\q^*, \la^*)$. Let $\bar{\q}_\kit$ be defined as \eqref{eq.runningaverages}. Let the graph $\mathcal{G}$ be connected (Assumption~\ref{ass.conn}). The following relations hold:
\begin{enumerate}
\item[\emph{(a)}] $\displaystyle 0\leq \La(\bar{\q}_\kit, \la^*) - \La(\q^*, \la^*) \leq \frac{C_0}{\kit+1}$,
\item[\emph{(b)}] $\displaystyle \lim_{\kit\to\infty}||{\q}^{(\kit)}-\q^*||\to 0$,
\end{enumerate}
where $C_0\geq 0$ is a constant that depends on the distance of the initial guess to the optimal solution, i.e., $||{\q}^{(0)} - \q^*||$ and $||\la^{(0)}-\la^*||$, and on the parameter $\rho$.
\end{theorem}
\begin{proof}
Since, in the problem~\eqref{eq.esdpsimplifiedadmm}, the sets $\mathcal{Z}_{i,j}$ and $\mathcal{P}_i$ are closed and convex, and the costs $f_i$ are proper and convex, the part \emph{(a)} of the proof follows from \cite[Theorem 4.1]{He2011}. Since the constraint \eqref{coupleeq} defines a linear system with full-column rank, the part \emph{(b)} of the proof follows from \cite[Theorem 1]{Mota2013}.
\end{proof}

\begin{algorithm}[t]
\caption{Distributed ADMM Algorithm for Problem~\eqref{eq.esdpsimplifiedadmm}}\label{alg.ADMM}\footnotesize
Set $\y_{i,j}^{i\,(0)}$, $\p_i^{(0)}$, $\z_{i,j}^{(0)}$, $\la_{i,j}^{i\,(0)}$ to zero, for all the nodes \\
\textbf{Input:} $\y_{i,j}^{i\,(\kit)}$, $\p_i^{(\kit)}$, $\z_{i,j}^{(\kit)}$, $\la_{i,j}^{i\,(\kit)}$, for all $j\in\mathcal{N}_i$\\
\color{white}aa\color{black}\vrule \color{white}aa\color{black}\begin{minipage}{0.45\textwidth}
\textbf{1:} Each node update its local variables $\y_{i,j}^{i(\kit)}$, $\p_i^{(\kit)}$ to $\y_{i,j}^{i(\kit+1)}$, $\p_i^{(\kit+1)}$ by the convex program~\eqref{eq.problemfinal} up to a defined accuracy ${\varepsilon}$\\ 
\textbf{2:} Each node sends its local vector $\y_{i,j}^{i\,(\kit+1)}$ to its neighbor $j$, for all $j\in\mathcal{N}_i$ \\
\textbf{3:} Each node computes, for all $j\in\mathcal{N}_i$ its $\z_{i,j}^{(\kit+1)}$ via Equation~\eqref{eq.z.itf}\\
\textbf{4:} Each node computes, for all $j\in\mathcal{N}_i$ 
\begin{equation*}\label{la.it}
\la_{i,j}^{i\,(\kit+1)} = \la_{i,j}^{i\,(\kit)} + \rho (\y_{i,j}^{i\,(\kit+1)} - \z_{i,j}^{(\kit+1)})
\end{equation*}
\end{minipage}\\
\textbf{Output:} $\y_{i,j}^{i\,(\kit+1)}$, $\p_i^{(\kit+1)}$, $\z_{i,j}^{(\kit+1)}$, $\la_{i,j}^{i\,(\kit+1)}$, for all $j\in\mathcal{N}_i$
\end{algorithm}

Let now local Lagrangian functions be 
\begin{multline*}
\La_i(\y_{i}^{i}, \p_i,  \z_{i}, \la_{i}^i):=\\f_i(\y_i^i,\p_i) + \sum_{j\in\mathcal{N}_i} \left[\la_{i,j}^{i\transp}\yiji + \frac{\rho}{2}\left\|\yiji-\z_{i,j} \right\|_2^2 \right],
\end{multline*}
where $\z_{i}$ and $\la_i^i$ are the stacked vectors of the $\z_{i,j}$'s and $\la_{i,j}^i$ for $j \in \mathcal{N}_i$, respectively\footnote{
By these definitions, the total Lagrangian~\eqref{reglagr} can now be written as $\La(\y, \p,  \z, \la) = \sum_{i\in\mathcal{V}} \La_i(\y_{i}^{i}, \p_i,  \z_{i}, \la_{i}^i)$, and the update~\eqref{eq.argmin1} reads $(\y_{i}^{i\,(\kit+1)}, \p_i^{(\kit+1)}) = \arg\min_{\yiji \in \mathcal{Z}_{i,j}, \p_i\in\mathcal{P}_i} \left\{\La_i(\y_{i}^{i}, \p_i,  \z_{i}^{(\kit)}, \la_{i}^{i (\kit)}) \right\}.$
}. Assume we are interested in determining how many iterations $t$ are needed to reach a given average local accuracy $\eta \geq 0$, meaning, 
\begin{equation}
0\leq \frac{1}{n}\sum_{i\in\mathcal{V}}\left(\La_i(\bar{\q}_{i,\kit}, \la_{i}^{i*}) - \La_i({\q}_{i}^*, \la_{i}^{i*})\right) \leq \eta,\label{curly}
\end{equation}
where $\bar{\q}_{i,\kit}$ is the running average (as in~\eqref{eq.runningaverages}) of the local vector $\q_i^{(\kit)}:=(\y_{i}^{i(\kit)\transp}, \p_i^{(\kit)\transp},  \z_{i}^{(\kit)\transp})^\transp$. The following proposition holds. 

\begin{proposition}\label{prop.time}
Let $\q_i^{(\kit)}$ be the local solution generated with Algorithm~\ref{alg.ADMM} applied to the convex problem~\eqref{eq.esdpsimplifiedadmm}, whose solution for sensor node $i$ is denoted by $(\q_i^*, \la_i^{i*})$. Let $\bar{\q}_{i,\kit}$ be defined as \eqref{eq.runningaverages} for $\q_i^{(\kit)}$. Let the graph $\mathcal{G}$ be connected (Assumption~\ref{ass.conn}). Let Algorithm~~\ref{alg.ADMM} be initialized with $\z^{(0)} = \mathbf{0}$ and $\la^{(0)} = \mathbf{0}$. 
Let $\eta$ be a given average local accuracy level, as expressed in~\eqref{curly}. If the number of iterations $\kit$ is chosen as
\begin{equation*}
\kit\geq \kit_{\eta} := \max_{i\in\mathcal{V}}\left\lceil \frac{1}{2 \rho \eta}(\rho^2 ||\z_i^{*}||_2^2 + ||\la_i^{i *}||_2^2) +1 \right\rceil,
\end{equation*}
where $\lceil\cdot\rceil$ represents the ceiling operator, then the accuracy $\eta$ is reached. 
\end{proposition}
\begin{proof}
The proof follows from Theorem~\ref{theo.convergence}. From point \emph{(a)} of Theorem~\ref{theo.convergence},
\begin{multline}
0\leq \La(\bar{\q}_\kit, \la^*) - \La(\q^*, \la^*)  =\\  \sum_{i\in\mathcal{V}}\left( \La_i(\bar{\q}_{i,\kit}, \la_{i}^{i*}) - \La_i({\q}_{i}^*, \la_{i}^{i*}) \right) \leq\frac{C_0}{\kit+1}. \label{theodummy1}
\end{multline}
From \cite[Theorem 4.1]{He2011}, $C_0$ can be expressed as 
\begin{multline}
C_0 = (2\rho ||\z^{(0)} - \z^*||_2^2 + \rho^{-1}||\la^{(0)}-\la^*||_2^2 )/2 = \\
\sum_{i\in\mathcal{V}}(\rho ||\z_{i}^{*}||_2^2 + \rho^{-1}||\la_i^{i*}||_2^2 )/2,\label{theodummy2}
\end{multline}
where the last simplification is due to the initialization of $\z^{(0)}$ and $\la^{(0)}$ at zero, and $\z_{i,j}=\z_{j,i}$. Combining \eqref{theodummy1}, \eqref{theodummy2}, and \eqref{curly} we obtain
\begin{multline*}
\frac{1}{n}\sum_{i\in\mathcal{V}}\left(\La_i(\bar{\q}_{i,\kit}, \la_{i}^{i*}) - \La_i({\q}_{i}^*, \la_{i}^{i*})\right)\leq\\  \frac{\max_{i\in\mathcal{V}}\{\rho^2 ||\z_{i}^{*}||_2^2 + ||\la_i^{i*}||_2^2\}}{2\rho (t+1)} = \eta
\end{multline*}
from which the claim follows. 
\end{proof}
Proposition~\ref{prop.time} says that the number of iterations for a given average local accuracy does not depend on the network size, but only on the worst local initial error. We can also characterize the total communication cost for node $i$ to reach a given accuracy level (which also does not depend on the network size) as follows. 
\begin{corollary}\label{coro.comm}
Under the same premises of Proposition~\ref{prop.time}, the communication cost ${c}_i$ for sensor node $i$ (i.e., the number of scalar numbers to send) to reach a desired average local accuracy $\eta$ is lower bounded by $c_i \geq 9|\mathcal{N}_i| \kit_{\eta}.$
\end{corollary}
\begin{proof}
Straightforward given the communication cost counting of Section~\ref{sec:comparisonadmm} and Proposition~\ref{prop.time}.
\end{proof}
\vskip0.2cm

Theorem~\ref{theo.convergence} indicates an $O(1/\kit)$ rate of convergence of Algorithm~\ref{alg.ADMM} in ergodic sense (i.e., in the sense of the running average vector).  We note that this $O(1/\kit)$ convergence is \emph{fast} if one looks at the very loose assumptions. As a matter of fact, $f$ could also have been non-differentiable; we report that typically non-differentiable problems solved using sub-gradient algorithms converge as $O(1/\sqrt{\kit})$ \cite{Duchi2012}.

The mentioned $O(1/\kit)$ convergence rate assumes perfect and synchronous communication at step 4) and that the optimizations at steps 3) and 5) are carried out exactly. In real situations, these are rather restrictive requirements. In practice, communication is affected by noise~\cite{Schizas2008}, packages can be dropped, and it is in general asynchronous among the sensor nodes. In addition, the often limited computational capabilities of the sensor nodes limit the possibility to obtain highly accurate solutions for the SDP in step 3). The strength of ADMM is however to be resilient to these issues, which in turn means that ADMM can be employed and convergence can be guaranteed also with these issues present~\cite{Boyd2011}. In this paper, we decided to focus on the loss of synchronicity and limited computation capabilities problems, since we believe they are the most critical ones in our application. 

\subsection{Properties of Algorithm~\ref{alg.ADMM} (Asynchronous)}

First of all, we consider asynchronous communication. Following the main bulk of research in ADMM, we consider an edge set perspective. Suppose that at iteration $\kit$ only a subset of all the existing links is activated, denoted by $\mathcal{E}^{(\kit)}$, and suppose that Algorithm~\ref{alg.ADMM} is run in an asynchronous fashion, where at each iteration $\kit$ we only consider the variables associated with $\mathcal{E}^{(\kit)}$ and we communicate only through $\mathcal{E}^{(\kit)}$. At each iteration $\kit$, we let the symmetric adjacency matrix associated with $\mathcal{E}^{(\kit)}$ be denoted as $\mathbf{A}^{(\kit)}$. We further assume the following. 

\begin{assumption}\label{ass.ber}
At each iteration $\kit$ the symmetric adjacency matrix $\mathbf{A}^{(\kit)}$ is generated by an i.i.d. Bernoulli process with $\mathrm{Pr}[[\mathbf{A}^{(\kit)}]_{ij} = 1]=s_{ij}>0$ for all $(i,j)\in\mathcal{E}$, with a given probability $0< s_{ij}\leq 1$.
\end{assumption}

\begin{assumption}\label{ass.union}
Let $\mathcal{G}^{(\kit)}:=(\mathcal{V},\mathcal{E}^{(\kit)})$. For every $\kit'\geq 0$, there exists an integer $T>0$ such that: 
\begin{enumerate}
\item[\emph{(i)}] the union of the edge sets satisfies $\bigcup_{\ell=\kit'}^{\kit'+T} \mathcal{E}^{(\ell)} = \mathcal{E} $;
\item[\emph{(ii)}] the union graph, i.e., $\bigcup_{\ell=\kit'}^{\kit'+T}  \mathcal{G}^{(\ell)}$, is connected. 
\end{enumerate}
\end{assumption}

These assumptions are rather standard in stochastic distributed optimization~\cite{Duchi2012, Iutzeler2013}. The convergence of Algorithm~\ref{alg.ADMM} under asynchronous communication can now be formally stated as follows.  

\begin{theorem}\label{theo.asynch}
Let $\q^{(\kit),\textrm{asy}} = ({\y}^{(\kit) \transp}, {\p}^{(\kit) \transp}, \z^{(\kit) \transp})^\transp$ be the solution generated by Algorithm~\ref{alg.ADMM} run in an asynchronous fashion, where at each iteration only a subset of edges are active. Let $(\q^*, \la^*)$ be the solution of the convex problem~\eqref{eq.esdpsimplifiedadmm}. Under Assumptions~\ref{ass.ber} and \ref{ass.union},
$$
\lim_{\kit\to\infty}\left\|\q^{(\kit),\textrm{asy}} - \q^*\right\| \to 0,\quad \textrm{almost surely}.
$$
\end{theorem}
\begin{proof}
The proof is an application of~\cite[Theorem 3 and Lemma 4]{Iutzeler2013}. Consider~\cite[Theorem 3]{Iutzeler2013}: Assumption 1 is valid since in the problem~\eqref{eq.esdpsimplifiedadmm}, the sets $\mathcal{Z}_{i,j}$ and $\mathcal{P}_i$ are closed and convex, and the costs $f_i$ are proper and convex. The Assumptions~2 and 3 are our Assumptions~\ref{ass.ber} and \ref{ass.union}. Problem~\eqref{eq.esdpsimplifiedadmm} can be put as the non-smooth unconstrained problem~(2) of~\cite{Iutzeler2013} and its dual is the problem~(12) of \cite{Iutzeler2013}. With this in place, by \cite[Theorem 3]{Iutzeler2013} we have now almost sure convergence in the dual domain for Algorithm~\ref{alg.ADMM}. By \cite[Lemma 4]{Iutzeler2013} primal convergence follows, after which the claim is proven. 
\end{proof}

\subsection{Properties of Algorithm~\ref{alg.ADMM} (Computation errors)}

\begin{table*}
\centering
\caption{Analytical comparison of the available distributed algorithms. Both SGO and ADMM can be applied to the E-ML formulation.}
\label{Table.comparison}
\setlength{\extrarowheight}{3pt}
\vskip-.5cm
\begin{tabular}{lccc}
& & & \\
\toprule
& E-ML with ADMM &  SGO of \cite{Shi2010} & MVU of \cite{Simonetto2013a} \\ \hline 
Size of the Convex Problem & $ 7 |\mathcal{N}_i|+2 |\mathcal{N}_{i,\mathrm{a}}|+3$ & $4 |\mathcal{N}_i|+3$ & not applicable  \\
Computational Complexity &  $O\left((|\mathcal{N}_i|+|\mathcal{N}_{i,\mathrm{a}}|)^3\right)$& $ O\left(|\mathcal{N}_i|^3\right) $ & $O\left(|\mathcal{N}_i|^2\right)$ \\
Communication cost & $9 |\mathcal{N}_i|$ & $2|\mathcal{N}_i|$ & $O(|\mathcal{N}_i|)$\\
Type of distributed algorithm & Parallel, ADMM  & Sequential, Gauss-Seidel & Parallel, Primal-Dual Subgradient and Consensus \\
Convergence rate & $O(1/\kit)$, (Ergodic) & $ O(r^{\kit/n})$, (Actual), and $O(n/\kit)$, (Ergodic)  & $O(\tau_\textrm{mix} \log^2(n)/\kit)$, (Ergodic) \\
\bottomrule
\end{tabular}
\end{table*}

The second aspect that we consider is the limited computation capabilities of the sensor nodes. In particular, we assume that each of the subproblems~\eqref{eq.argmin1} is solved up to an accuracy $\varepsilon$, i.e., the optimal solution of each subproblem satisfies
\begin{multline}
\hskip-0.2cm0\leq \La_i(\y_{i}^{i},\p_i, \z_i^{(\kit)}\hskip-0.1cm, \la_i^{i(\kit)})-\La_i(\y_{i}^{i(\kit+1)}\hskip-0.2cm,\p_i^{(\kit+1)}\hskip-0.2cm, \z_i^{(\kit)}\hskip-0.1cm, \la_i^{i(\kit)})\leq \varepsilon, \\ \textrm{for all } \y_{i,j}^{i}\in\mathcal{Z}_{i,j},\p_i\in\mathcal{P}_i.
\label{eq.eps}
\end{multline}
The following theorem is now in place.
\begin{theorem}\label{prop.2}
Let $\bar{\q}_{\kit,\varepsilon}$ be the running solution generated with Algorithm~\ref{alg.ADMM} under the assumption that each of the subproblems~\eqref{eq.argmin1} is solved up to an accuracy $\varepsilon$, as specified by condition~\eqref{eq.eps}. Let $(\q^*, \la^*)$ be the solution of the convex problem~\eqref{eq.esdpsimplifiedadmm}. Then the following holds: 
$$
0\leq \La(\bar{\q}_{\kit,\varepsilon}, \la^*) - \La(\q^*, \la^*) \leq \displaystyle\frac{C_0}{\kit+1} + n\varepsilon,
$$
where $C_0\geq 0$ is a constant that depends on the distance of the initial guess to the optimal solution, i.e., $||{\q}^{(0)} - \q^*||$ and $||\la^{(0)}-\la^*||$, and on the parameter $\rho$. 
\end{theorem}
\begin{proof}
The proof follows directly from~\cite{He2011} substituting their (3.5) with our~\eqref{eq.eps}. 
\end{proof}

Proposition~\ref{prop.2} implies that Algorithm~\ref{alg.ADMM} converges as $O(1/\kit)$ to an error floor with magnitude $n\varepsilon$. 

\subsection{Comparison of Algorithm~\ref{alg.ADMM} with Alternatives}\label{sec:comparisonadmm}

In this section, we analyze the computational complexity and the communication cost of Algorithm~\ref{alg.ADMM} and we compare it to some other distributed algorithms for convex relaxations, namely the sequential greedy optimization (SGO) algorithm of~\cite{Shi2010} and the distributed maximum variance unfolding (MVU) algorithm of~\cite{Simonetto2013a} (we leave out the approach of~\cite{Srirangarajan2008} since convergence has not been proven). The aim is to show the added value in using ADMM especially for large-scale networks. For simplicity, the nodes are located in $\mathbb{R}^2$. 

\textbf{E-ML with ADMM (Algorithm~\ref{alg.ADMM})} At each iteration, for each sensor node, the most complex operation is to solve the convex program~\eqref{eq.problemfinal}. This convex program optimizes over $\y_i^i, \p_i, \mathbold{\gamma}$ and it comprises of $7|\mathcal{N}_i| + 2|\mathcal{N}_{i,\mathrm{a}}|+3$ scalar variables, $3|\mathcal{N}_i| + 2|\mathcal{N}_{i,\mathrm{a}}|$ scalar equality/inequality constraints, and $4 |\mathcal{N}_i| + 2|\mathcal{N}_{i,\mathrm{a}}|$ LMI constraints of size at most $10 \times 10$ (which is represented by the LMI with $\gamma_{i,j}$)%
%
\footnote{In fact, assuming $\x_i\in\mathbb{R}^2$, then $\y_i^i\in\mathbb{R}^{6|\mathcal{N}_i|+3}$, $\p_i\in\mathbb{R}^{2|\mathcal{N}_{i,\mathrm{a}}|+3}$, $\mathbold{\gamma}\in\mathbb{R}^{|\mathcal{N}_i|}$, and eliminating the overlapping variables between $\y_i^i$ and $\p_i$ the count follows. Furthermore, the local optimization has a part of \eqref{eq.cons.10b} and \eqref{eq.gamma2} as equality/inequality constraints, in total $3|\mathcal{N}_i|+ 2|\mathcal{N}_{i,\mathrm{a}}|$, and the other part of \eqref{eq.cons.10b} plus \eqref{c.lmi3} and \eqref{eq.gamma1} as LMI, in total $4|\mathcal{N}_i|+ 2|\mathcal{N}_{i,\mathrm{a}}|$ LMIs of dimension at most $10 \times 10$ in the case of \eqref{eq.gamma1}.}. %
%
This yields a computational complexity of at least $O((|\mathcal{N}_i| + |\mathcal{N}_{i,\mathrm{a}}|)^3)$ (see~\cite{Boyd2004a} for details on operation counts). The communication cost per iteration per sensor is proportional to the number of scalar variables that sensors have to send, and each sensor has to send the updated $\y_{i,j}^j\in \mathbb{R}^{9}$ to its neighbor $j$, for each neighbor, i.e., a communication cost of $9|\mathcal{N}_i|$.

\textbf{SGO.} The SGO algorithm of~\cite{Shi2010} is sequential in nature, meaning that only one local optimization can be run at the time, and although its convergence has been argued, no formal proof has been given for the convergence rate\footnote{Coloring procedures as in~\cite{Bertsekas1997} could be employed to partially parallelize SGO. These coloring techniques depend on the availability of a coloring scheme before running the SGO. Coloring schemes are NP-hard problems, and albeit there are decentralized techniques to compute bounds, the number of iterations to achieve a given accuracy is between $O(\log(n))$ and $O(n \exp(n))$ \cite{Kuhn2006, Duffy2008}, which undermines their applicability for large-scale settings.}. Furthermore noise-resilience as well as computation error-resilience are unknown features of SGO. For ease of comparison, we consider the range-based localization SGO and ignore the anchors for simplicity. In this context, and in the case we apply SGO to our E-ML formulation, at each iteration, for the one active sensor node (given that we keep $\x_j$ fixed and there is no $\z_{i,j}$ or $\mathbold{\gamma}$ variable) the most complex operation is to solve a convex program comprising of $4|\mathcal{N}_i| + 3$ variables and $3 |\mathcal{N}_{i}|$ LMI constraints, which leads to a computational complexity of at least $O(|\mathcal{N}_{i}|^3)$ (see also~\cite{Shi2010}). The communication costs per iteration for the one active sensor is proportional to the number of scalar variables that have to be sent (the updated $\x_i$) multiplied by the number of sensor nodes they have to be sent to (the neighbors), yielding a cost of $2|\mathcal{N}_i|$. For the convergence rate, the best convergence rate that we can expect from a Gauss-Seidel algorithm (with some strong assumptions on the constraints and cost function) is linear~\cite{Bertsekas1997}, i.e., the convergence rate is $O(r^\kit)$ for a certain (problem-dependent and a priori unknown) $0<r<1$. Given that one iteration of the Gauss-Seidel comprises $n$ sub-iterations of the SGO, the convergence rate of SGO is at best $O(r^{\kit/n})$, or $O(n/\kit)$ in ergodic sense. 

\textbf{MVU.} The MVU algorithm of~\cite{Simonetto2013a} is parallel in nature, employs a primal-dual scheme with a nested consensus step, and cannot handle anchors. It is based on the decentralized spectral decomposition algorithm of~\cite{Kempe2008} and it requires each node to eventually locate all the others. At each iteration, for each sensor node, the computational complexity is at most $O(|\mathcal{N}_i|^2)$ and the communication cost at most $O(|\mathcal{N}_i|)$. The convergence rate is based on the convergence of the decentralized spectral decomposition algorithm, which requires $O(\tau_\textrm{mix} \log^2(n))$ sub-iterations ($\tau_\textrm{mix}$ is the mixing time of a random walk on the graph $\mathcal{G}$), and on the convergence rate of the primal-dual scheme, proven to be $O(1/\kit)$ in ergodic sense. 

Table~\ref{Table.comparison} collects the performed analyses and indicates that ADMM may be the best choice to increase the convergence rate, especially in the case of large-scale networks. This comes with a limited increase in communication cost, which however can always be tuned choosing the neighborhood's size. In the next section, we display what this means in simulation results along with other relevant comparisons. 

\section{Numerical Simulations}\label{sec:num}

In this section, we report several numerical comparisons for both the centralized formulation~\eqref{eq.esdp} and the distributed Algorithm~\ref{alg.ADMM}. The aim of the section is to show how the E-ML relaxation performs under various noise conditions, to support the idea that tighter relaxations perform better in terms of position error (even though they may model the noise PDF wrongly), and to display the numerical properties of the distributed Algorithm~\ref{alg.ADMM}.  

\subsection{Centralized simulations}

We consider $2$-dimensional problems and we use the benchmark \texttt{test10-500} available online at \url{http://www.stanford.edu/~yyye/}, where the sensor nodes are randomly distributed in the unit box $[-0.5, 0.5]^2$. We let $\xi_{i,\ell}$ be the position error of sensor node $i$ for a certain realization of the noise $\ell$, i.e., $\xi_{i,\ell}:=||\hat{\x}_{i,\ell}-\x_i||_2$, where $\hat{\x}_{i,\ell}$ is the estimated position and $\x_i$ is the true position. We consider the position root mean square error ($\mathrm{PRMSE}$) as a metric of performance for the proposed convex relaxations, i.e., 
\begin{equation}
\mathrm{PRMSE} = \sqrt{\frac{\sum_{\ell=1}^L \sum_{i\in\mathcal{V}} \xi_{i,\ell}^2}{L}},
\end{equation}
where $L$ is the total number of noise realizations. Along with this metric, we consider the worst case maximum error, i.e., 
\begin{equation}
\mathrm{ME} = \max_{i\in\mathcal{V}, \ell\in[1,L]} \xi_{i,\ell},
\end{equation}
and we compute the Cram\'er-Rao lower bound (CRLB) as a comparison benchmark as in~\cite{Patwari2005}. \vskip0.2cm

\textbf{Gaussian noise setting.} (Figures~\ref{fig.comparison} and \ref{fig.comparison2}) In the first example, we focus on a Gaussian noise setting. We fix the maximum number of neighbors for each sensor node to $3$, we set the number of anchors to $m=5$, we consider additive white noise of the same standard deviation $\sigma_{i,j} = \sigma_{i,k,\mathrm{a}}$ for all the measurements, and we average over $50$ realizations. 

In Figure~\ref{fig.comparison}, we compare the E-ML approach, i.e., the problem~\eqref{eq.esdp} with cost function $f_{\mathrm{GN}}(\mathbold{\delta},\mathbold{\epsilon},\d,\e)$, with the ESDP relaxation of~\cite{Wang2008} (considered to be the state-of-the-art in convex relaxations) by increasing the number of sensor nodes $n$ and keeping all the other parameters the same ($\sigma_{i,j} = \sigma_{i,k,\mathrm{a}} = 0.1$). As we can see, the performance of E-ML is better than the one of ESDP, albeit only slightly. Furthermore, as one could expect, by increasing $n$ we average out the noise, which in turn means a better average performance and less difference among the two schemes. 

In Figure~\ref{fig.comparison2}, we study the performance of the E-ML approach and of the ESDP relaxation by increasing the noise value, for $n = 8$. As we can see, the performance of E-ML is again slightly better than the one of ESDP, and the difference increases with the noise value (notice that the graph is in logarithmic scale). 

 
\begin{figure}
\centering
\footnotesize
\psfrag{a}{\hskip-1.4cm Number of sensor nodes $n$}
\psfrag{b}{\hskip-0.8cm Position Error [-]}
\psfrag{c}{ESDP \cite{Wang2008}, $\mathrm{ME}$}
\psfrag{d}{E-ML, problem~\eqref{eq.esdp}, $\mathrm{ME}$}
\psfrag{e}{$\sqrt{\textrm{CRLB}}/n$}
\psfrag{f}{ESDP \cite{Wang2008}, $\mathrm{PRMSE}/n$}
\psfrag{g}{E-ML, problem~\eqref{eq.esdp}, $\mathrm{PRMSE}/n$}
\includegraphics[width=0.475\textwidth]{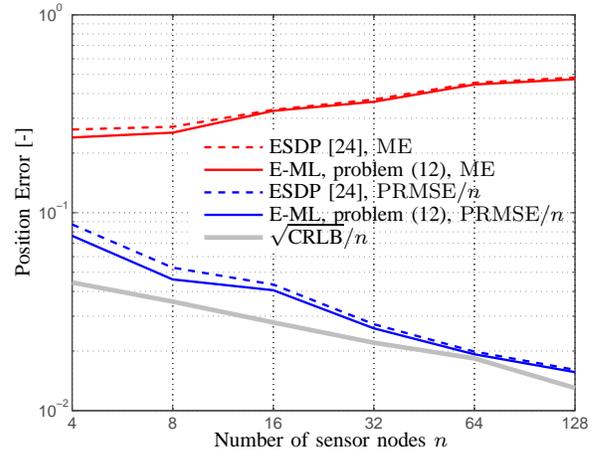}
\caption{Comparison between E-ML relaxation~\eqref{eq.esdp} and ESDP relaxation of~\cite{Wang2008} in the Gaussian noise setting for different values of the sensor node number $n$ and fixed $\sigma_{i,j} = \sigma_{i,k,\mathrm{a}} = 0.1$. }
\label{fig.comparison}
\end{figure}

\begin{figure}
\centering
\footnotesize
\psfrag{a}{\hskip-2.4cm Noise standard deviation $\sigma_{i,j} = \sigma_{i,k,\mathrm{a}}$ [-]}
\psfrag{b}{\hskip-0.8cm Position Error [-]}
\psfrag{c}{ESDP \cite{Wang2008}, $\mathrm{ME}$}
\psfrag{d}{E-ML, problem~\eqref{eq.esdp}, $\mathrm{ME}$}
\psfrag{e}{$\sqrt{\textrm{CRLB}}/n$}
\psfrag{f}{ESDP \cite{Wang2008}, $\mathrm{PRMSE}/n$}
\psfrag{g}{E-ML, problem~\eqref{eq.esdp}, $\mathrm{PRMSE}/n$}
\includegraphics[width=0.475\textwidth]{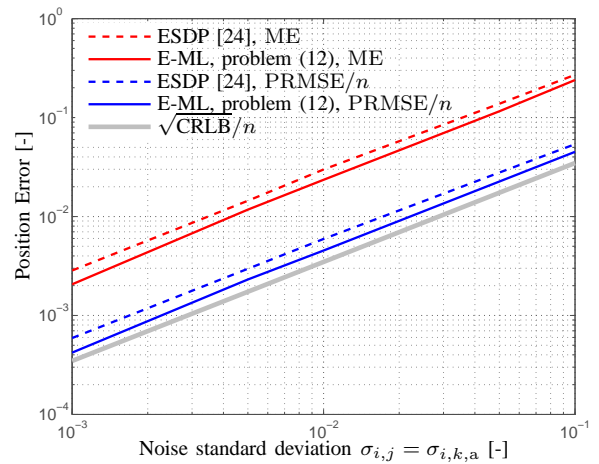}
\caption{Comparison between E-ML relaxation~\eqref{eq.esdp} and ESDP relaxation of~\cite{Wang2008} in the Gaussian noise setting for different values of the measurement noise standard deviation $\sigma_{i,j} = \sigma_{i,k,\mathrm{a}}$ and fixed $n = 8$. }
\label{fig.comparison2}
\end{figure} 

\begin{figure}
\centering
\footnotesize
\psfrag{a}{\hskip-2.4cm Noise standard deviation $\sigma_{i,j} = \sigma_{i,k,\mathrm{a}}$ [-]}
\psfrag{b}{\hskip-0.8cm $\mathrm{PRMSE}/n$ [-]}
\psfrag{c}{Lapl. E-ML~\eqref{eq.esdp}, $n = 8$}
\psfrag{d}{ESDP~\cite{Wang2008}, $n = 8$}
\psfrag{f}{GN E-ML~\eqref{eq.esdp}, $n = 8$}
\psfrag{g}{$\sqrt{\textrm{CRLB}}/n$, $n = 8$}
\psfrag{e}{Lapl. E-ML~\eqref{eq.esdp}, $n = 64$}
\psfrag{h}{ESDP~\cite{Wang2008}, $n = 64$}
\psfrag{i}{GN E-ML~\eqref{eq.esdp}, $n = 64$}
\psfrag{l}{$\sqrt{\textrm{CRLB}}/n$, $n = 64$}
\includegraphics[width=0.475\textwidth]{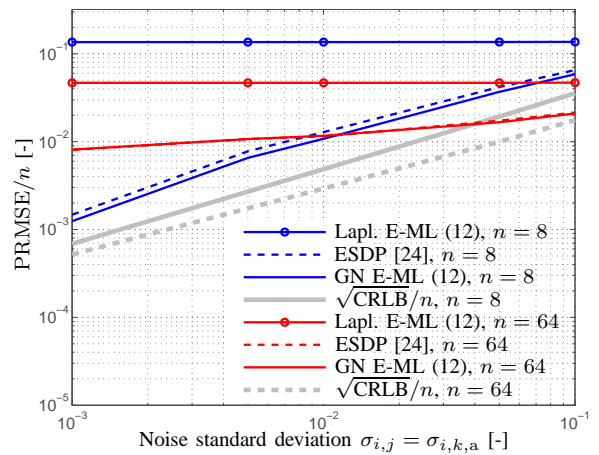}
\caption{Comparison between Laplacian E-ML relaxation, i.e., \eqref{eq.esdp} with $f_{\mathrm{L}}(\d,\e)$, ESDP relaxation of~\cite{Wang2008}, and Gaussian E-ML relaxation, i.e.,~\eqref{eq.esdp} with $f_{\mathrm{GN}}$, in the Laplacian noise setting for different values of the sensor node number $n$ and different noise values $\sigma_{i,j} = \sigma_{i,k,\mathrm{a}}$. }
\label{fig.comparison3}
\end{figure}
 
\vskip0.2cm
\textbf{Laplacian noise setting.} (Figure~\ref{fig.comparison3}) In this second example, we focus on a Laplacian noise setting. Also in this example, we fix the maximum number of neighbors for each sensor node to $3$, and we set $m=5$ and $L = 50$.

In Figure~\ref{fig.comparison3}, we compare the Laplacian E-ML relaxation, i.e., problem~\eqref{eq.esdp} with cost function $f_{\mathrm{L}}(\d,\e)$, the Gaussian E-ML relaxation, i.e., problem~\eqref{eq.esdp} with cost function $f_{\mathrm{GN}}(\mathbold{\delta},\mathbold{\epsilon},\d,\e)$, and the ESDP relaxation of~\cite{Wang2008}. We use the modified version of the CRLB of~\cite{Leng2010} as a benchmark, since Laplacian distributions are not differentiable. We vary both the number of sensor nodes $n$ and the noise value. As we can see, although the Laplacian E-ML relaxation correctly models the noise distribution, it performs worse than the other convex relaxations. The reason is that it is not derived from a rank-$D$ relaxation and therefore it is a ``looser'' relaxation with respect to the other ones considered in this example. 

\subsection{Distributed simulations}

We use the same setting of the centralized simulations (i.e., anchor number $m = 5$, and maximum number of neighbors for each sensor node is $3$) and we consider Gaussian noise. We test Algorithm~\ref{alg.ADMM} based on ADMM for different values of $n$, computation accuracy $n \varepsilon$, and asynchronous communication. In order to generate the computation error, we set \texttt{sedumi.eps} to $\varepsilon$, which is an upper bound\footnote{SeDuMi considers this tolerance to be related also to feasibility and not only optimality, as we do, see~\cite{Sturm1998} for details.} for our definition of $\varepsilon$. We set the regularization parameter $\rho = 0.3$. We first focus on synchronous communication and then on the asynchronous implementation. \vskip0.2cm

\textbf{Synchronous case.} (Figures~\ref{fig.distr1}, \ref{fig.distr2}, and \ref{fig.distr3})
Figures~\ref{fig.distr1}, \ref{fig.distr2}, and \ref{fig.distr3} collect the synchronous communication results for Algorithm~\ref{alg.ADMM} and confirm the $O(1/\kit)$ convergence of ADMM. 

In Figures~\ref{fig.distr1} and \ref{fig.distr2}, we fix the centralized problem as the relaxation~\eqref{eq.esdp} with cost function $f_{\mathrm{GN}}(\mathbold{\delta},\mathbold{\epsilon},\d,\e)$ and we compare the convergence of Algorithm~\ref{alg.ADMM} to the centralized solution (Theorem~\ref{theo.convergence}) with the one of SGO applied to the same centralized problem. We also show the effect of computation inaccuracies (Theorem~\ref{prop.2}) supporting our theoretical findings. As we can see, by comparing SGO with the ADMM approach, we notice the slower convergence of the former (for a large-scale setting) due to its sequential nature (in fact, in the case of SGO, at each iteration $\kit$ we update only one sensor node position). We see also that SGO is resilient to computation inaccuracies (at least in this simulation), it seems to have a linear type of convergence (as argued), and it may be a choice in case of small-size networks. Further studies are however necessary to certify the reliability of SGO to a broader class of scenarios.  

Figure~\ref{fig.distr3} represents the sensor node locations computed as the solution of Algorithm~\ref{alg.ADMM} for different iterations $\kit$. The algorithm is initialized with $\X^{(0)} = \textbf{0}$ and then run till $\kit = 400$. The ``trajectories'' of the running averaged variables $\bar{\X}_{\kit}$ (i.e., the position part of the $\bar{\q}_\kit$ vector~\eqref{eq.runningaverages}) as a function of the iteration number $\kit$ are displayed. As we can see, for $\kit = 400$, Algorithm~\ref{alg.ADMM} is practically converged onto the real sensor node locations. 

\begin{figure}
\centering
\footnotesize
\psfrag{a}{\hskip-0.5cm Iteration $\kit$}
\psfrag{b}{\hskip-1.5cm $\La({\q}^{(\kit)},\lambda^*) - \La(\q^*,\lambda^*)$}
\psfrag{d}{$1/\kit$ line}
\psfrag{c}{$n = 8$, exact}
\psfrag{e}{$n = 32$, $n \varepsilon = 0.02$}
\psfrag{f}{$n = 128$, $n \varepsilon = 0.02$}
\psfrag{g}{\textsc{sgo}, $n = 8$, exact}
\psfrag{h}{\textsc{sgo}, $n = 128$, $n \varepsilon = 0.02$}
\includegraphics[width=0.475\textwidth]{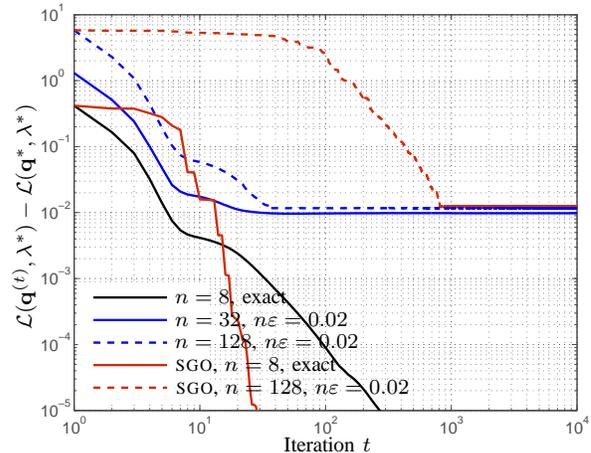}
\caption{Actual objective convergence of Algorithm~\ref{alg.ADMM} in different settings and comparison with SGO solving the same E-ML problem~\eqref{eq.esdp}.}
\label{fig.distr1}
\end{figure}

\begin{figure}
\centering
\footnotesize
\psfrag{a}{\hskip-0.5cm Iteration $\kit$}
\psfrag{b}{\hskip-1.5cm $\La(\bar{\q}_\kit,\lambda^*) - \La(\q^*,\lambda^*)$}
\psfrag{d}{$1/\kit$ line}
\psfrag{c}{$n = 8$, exact}
\psfrag{e}{$n = 32$, $n \varepsilon = 0.02$}
\psfrag{f}{$n = 128$, $n \varepsilon = 0.02$}
\psfrag{g}{\textsc{sgo}, $n = 8$, exact}
\psfrag{h}{\textsc{sgo}, $n = 128$, $n \varepsilon = 0.02$}
\includegraphics[width=0.475\textwidth]{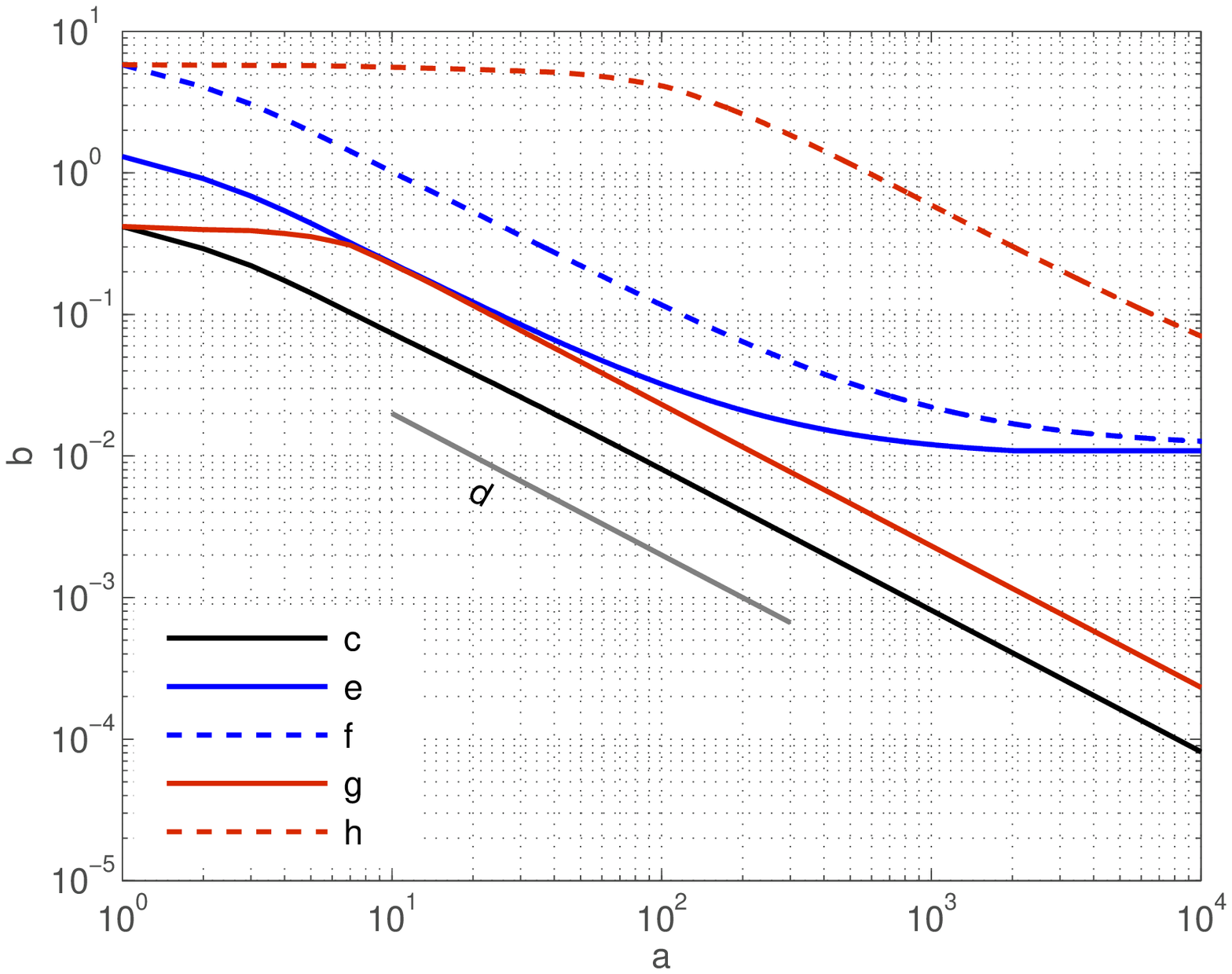}
\caption{Ergodic objective convergence of Algorithm~\ref{alg.ADMM} in different settings and comparison with SGO solving the same E-ML problem~\eqref{eq.esdp}.}
\label{fig.distr2}
\end{figure}

\begin{figure}
\centering
\footnotesize
\psfrag{a}{\hskip-0.2cm $x_1$}
\psfrag{b}{\hskip-0.2cm $x_2$}
\psfrag{c}{Anchors' positions}
\psfrag{d}{Sensor nodes' real positions}
\psfrag{e}{Distr. solution at iteration $\kit = 400$}
\psfrag{f}{Distr. solution from iteration $\kit = 0$ till $\kit = 400$}
\includegraphics[width=0.475\textwidth]{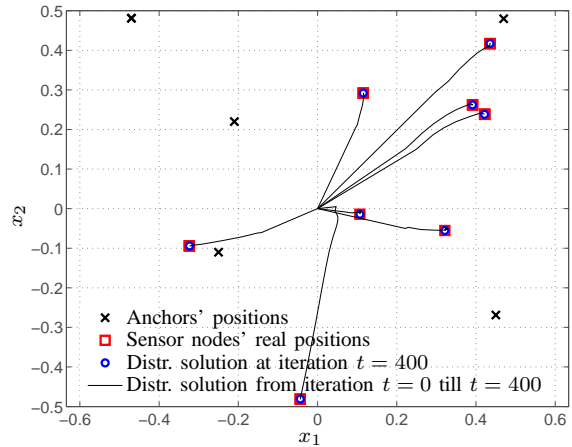}
\caption{Position solutions $\bar{\X}_\kit$ as a function of $\kit$ plotted as trajectories (using Algorithm~\ref{alg.ADMM} in its exact form) from $\kit = 0$ till $\kit = 400$. The values of $\bar{\X}_{400}$ (blue circles) are practically coincident with the real sensor node positions (red squares). The crosses represent the anchors.}
\label{fig.distr3}
\end{figure}

\vskip0.2cm
\textbf{Asynchronous case.} (Figure~\ref{fig.distr4}) For the asynchronous communication case, we use the same setting as the synchronous scenario and we consider different values for the number of sensor nodes $n$ and probability $s_{ij}$ (Assumption~\ref{ass.ber}). 

In Figure~\ref{fig.distr4}, the results are displayed. In particular, we have depicted the distance between the primal solution from Algorithm~\ref{alg.ADMM}, $\q^{(\kit)}$, and the optimal value found using the centralized problem~\eqref{eq.esdp}, i.e., $\q^*$. As we expect, Algorithm~\ref{alg.ADMM} converges to the optimal primal solution of the centralized problem~\eqref{eq.esdp} (Theorem~\ref{theo.asynch}). 

\begin{figure}
\centering
\footnotesize
\psfrag{a}{\hskip-0.5cm Iteration $\kit$}
\psfrag{b}{\hskip-0.4cm $||\q^{(\kit)}-\q^*||$}
\psfrag{c}{$n=8$, $s_{ij} = 1.00$}
\psfrag{d}{$n=8$, $s_{ij} = 0.75$}
\psfrag{e}{$n=8$, $s_{ij} = 0.25$}
\psfrag{f}{$n=32$, $s_{ij} = 0.25$}
\includegraphics[width=0.475\textwidth]{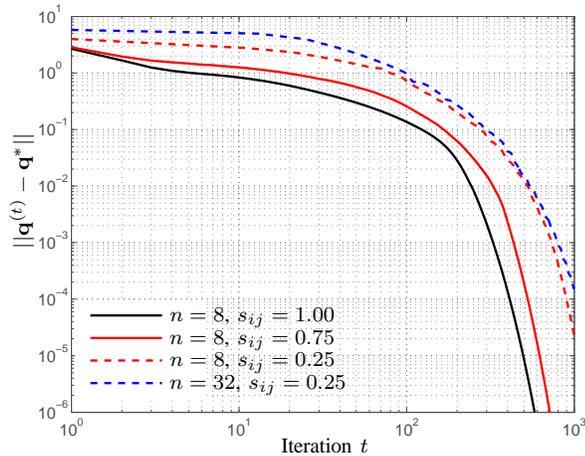}
\caption{Actual primal convergence of Algorithm~\ref{alg.ADMM} to the optimizer of the centralized E-ML problem~\eqref{eq.esdp} for asynchronous communication.}
\label{fig.distr4}
\end{figure}

\section{Conclusions}\label{sec:concl}

We have studied the sensor network localization problem. We have argued that employing convex relaxations based on a maximum likelihood formulation to massage the original non-convex formulation can offer a powerful handle on computing accurate solutions. In order to take full advantage of this aspect, we have shown that the relaxation has to be as tight as possible to the original non-convex problem, (in some cases, disregarding the noise model). Furthermore, we have discussed a distributed implementation of the resulting convex relaxation via the ADMM. By exploiting the analytical properties of ADMM (convergence rate, asynchronism-resilience, computation error-resilience), we have studied the resulting distributed algorithm showing its added value with respect to available techniques, especially in large networks. 

Among future research plans, we are interested in studying mobile sensor network localization problems by using convex relaxations based on a maximum a posteriori formulation of the estimation problem.

\ifCLASSOPTIONcaptionsoff
  \newpage
\fi

\bibliographystyle{ieeetr}
\bibliography{../../PaperCollection2}

\end{document}

%% file: package_mode.tex
\usepackage{amsmath,amsfonts,amssymb,fixmath}
\usepackage{color}
\usepackage{algorithm}

\usepackage{algorithmic}
\usepackage{comment}
\usepackage{multirow,array}
\def\bbm[#1]{\mbox{\boldmath $#1$}}      
\def\mpc[#1]{\marginpar{\color[rgb]{0.07,0.66,0.11} #1}}
\def\mpcp[#1]{\marginpar{\color[rgb]{1.00,0.50,0.00} \footnotesize{#1}}}

\def \X{\mathbf{X}}\def \kit{t}
\def \x{\mathbf{x}}\def \z{\mathbf{z}}

\def \Y{\mathbf{Y}}
\def\yiji{\y_{i,j}^i}
\def\yijj{\y_{i,j}^j}
\def \a{\mathbf{a}}\def \y{\mathbf{y}}
\def \d{\mathbf{d}}
\def\A{\mathbf{A}}
\def \e{\mathbf{e}}\def\p{\mathbf{p}}\def\w{\mathbf{w}}\def\q{\mathbf{q}}
 \def\la{\mathbold{\lambda}}

\def\nnorm[#1, #2]{\mbox{$\left. \vline\,\vline\,\displaystyle #1 \,\vline\,\vline \right._{#2}\,$}}
\def\La{\mathcal{L}}

\DeclareMathOperator*{\minimize}{\mathrm{minimize}}

\newtheorem{theorem}{Theorem}

\newtheorem{proposition}{Proposition}
\newtheorem{assumption}{Assumption}

\newtheorem{remark}{Remark} 

%% file: package_IEEE.tex
\usepackage{cite}

\ifCLASSINFOpdf
\else
\fi